\documentclass[10pt,twocolumn,journal]{IEEEtran}
\usepackage{amsmath}
\usepackage{amssymb}
\usepackage{amsfonts}
\usepackage{mathrsfs}
\usepackage{amsthm}
\usepackage{mathtools}
\usepackage{graphicx}
\usepackage[T1]{fontenc}
\usepackage{bm}
\usepackage{bbm}
\usepackage{scalerel}
\usepackage{esvect}
\usepackage{algorithm}
\usepackage{algpseudocode}
\usepackage{enumitem}
\usepackage{xcolor}
\usepackage[
  colorlinks=true, 
  linkcolor=blue,  
  urlcolor=black,  
  bookmarks=true,
] {hyperref}

\renewcommand{\v}[1]{\bm{#1}}
\renewcommand{\hat}[1]{\widehat{#1}}
\renewcommand{\tilde}[1]{\widetilde{#1}}
\newcommand{\norm}[1]{\left\lVert#1\right\rVert}
\newcommand{\abs}[1]{\left\lvert#1\right\rvert}
\newcommand{\ip}[2]{\left\langle#1,#2\right\rangle}

\newcommand{\ps}[1]{\left(#1\right)}

\newcommand{\reals}{\mathbb{R}}

\newcommand{\bigO}[1]{\mathcal{O}\ps{#1}}

\newcommand{\argmin}[1]{\underset{#1}{\arg \min} \enskip}
\newcommand{\argmax}[1]{\underset{#1}{\arg \max} \enskip}

\newcommand{\W}{\mathcal{W}}
\newcommand{\V}{\mathcal{V}}
\newcommand{\Wbeck}{\widetilde{\mathcal{W}}}
\newcommand{\Vbeck}{\widetilde{\mathcal{V}}}
\newcommand{\U}{\mathcal{U}}
\newcommand{\ones}{\mathbbm{1}}
\newcommand{\ind}[1]{\iota_{#1}}
\newcommand{\itercur}{^{\scaleto{k}{5pt}}}
\newcommand{\iternext}{^{\scaleto{k+1}{5pt}}}

\renewcommand{\vec}[1]{\vv{\bm{#1}}}

\newcommand{\proximal}{\mathrm{prox}}
\newcommand{\shrink}{\mathrm{shrink}}

\DeclareMathOperator*{\rank}{rank}
\DeclareMathOperator{\divergence}{div}
\DeclareMathOperator*{\sign}{\text{sign}}
\newtheorem{theorem}{Theorem}
\let\div\relax
\newcommand{\div}{\v{D}}
\let\vec\relax
\DeclareMathOperator{\vec}{vec}
\newcommand{\alt}[1]{\tilde{\v{#1}}}

\begin{document}
\title{Parallel Unbalanced Optimal Transport Regularization for Large-scale Imaging Problems}
\author
{
\IEEEauthorblockN{John~Lee,
                  Nicholas~P.~Bertrand,
                  Christopher~J.~Rozell,~\IEEEmembership{Senior Member,~IEEE,}}
\thanks{J.L. is with DSO National Laboratories, 12 Science Park Drive, Singapore 118225.
N.P.B. and  C.J.R. are with the School of Electrical and Computer Engineering, Georgia Institute of Technology, Atlanta, GA, 30332 USA.
Email: lzhanyi@dso.org.sg,
       nbertrand@gatech.edu,
       crozell@gatech.edu.
       This work was supported in part by NSF grant CCF-1409422, James S. McDonnell Foundation grant 220020399, and DSO National Laboratories of Singapore.}
}
\maketitle

\begin{abstract}
The modeling of phenomenological structure is a crucial aspect in inverse imaging problems.
One emerging modeling tool in computer vision is the optimal transport framework.
Its ability to model geometric displacements across an image's support gives it attractive qualities similar to optical flow methods that are effective at capturing visual motion, but are restricted to operate in significantly smaller state-spaces.
Despite this advantage, two major drawbacks make it unsuitable for general deployment:
(i) it suffers from exorbitant computational costs due to a quadratic optimization-variable complexity, and
(ii) it has a mass-balancing assumption that limits applications with natural images.
We tackle these issues simultaneously by proposing a novel formulation for an \emph{unbalanced} optimal transport regularizer that has \emph{linear} optimization-variable complexity.
In addition, we present a general parallelizable proximal method for this regularizer, and demonstrate superior empirical performance on novel dynamical tracking applications in synthetic and real video.
\end{abstract}

\begin{IEEEkeywords}
Inverse problems, optimal transport, proximal methods, robust principal components analysis
\end{IEEEkeywords}

\IEEEpeerreviewmaketitle

\section{Introduction}

\IEEEPARstart{W}{hile} there has been much progress in inverse imaging fueled by advances in numerical optimization \cite{afonso2010fast, afonso2010augmented, chambolle2011first} and theoretical recovery guarantees \cite{tibshirani1996regression, candes2006robust, candes2011robust}, the modeling of physical phenomena remains a crucial aspect in solving such problems.
Many inverse imaging problems (e.g., denoising, deconvolution, inpainting) are cast as optimization programs and are solved by exploiting \emph{a priori} structure through the construction of meaningful modeling regularizers.
Regularizers that employ $\ell_p$ metrics are often useful at modeling such phenomena because they efficiently describe point-wise statistics under fixed grid representations (e.g., rectangular pixel layout in imaging).
However, $\ell_p$ metrics are sometimes inadequate because they fundamentally lack the mechanism to capture geometric relationships between support coordinates.
In contrast, the optimal transport (OT) problem is a framework that explicitly accounts for geometric relationships by modeling a signal as mass that incurs a cost to move around its support.
Under certain geodesic metrics (e.g. Euclidean notions of geometry in the support), the OT framework (e.g., Wasserstein distance) very naturally quantifies uncertainty and deformation.
As a result of these attractive properties, computational approaches to regularization with OT have found a variety of imaging applications such as incompressible fluid flow \cite{jordan1998variational, peyre2015entropic}, temporal dynamics of sparse signals \cite{Charles2017emd, bertrand2018earth}, and physical deformation in medical images \cite{karlsson2017generalized}.

Despite OT's favorable modeling capabilities, two major drawbacks exist.
First, evaluating OT is traditionally computationally expensive, limiting the state-space size $N$ (e.g., number of pixels) in practice.
In the most general formulation, solving OT amounts to solving an LP with $N^2$ variables and even efficient solvers (e.g., interior-point or simplex methods) have a computational complexity of $\bigO{N^3 \log N}$.
The recent computational advances in OT are partially attributed to the introduction of \emph{Sinkhorn distances} \cite{cuturi2013sinkhorn} as an efficient approximation method that trades-off accuracy for efficiency, reducing the per iteration costs to $\bigO{N^2}$ (see \cite{peyre2019computational} and references therein).
Unfortunately, this trade off can be potentially counterproductive in applications that demand both accuracy and speed.
Under a more restrictive case where the underlying geodesic metric is assumed to be Euclidean, Beckmann's formulation \cite{beckmann1952continuous} offers an attractive linear variable complexity similar to Sinkhorn methods but without sacrificing accuracy.
The Euclidean geodesic restriction is arguably a reasonable modeling assumption in imaging and video, by virtue of the success of optical flow methods which hold similar assumptions.
Recently, Beckmann's OT formulation has been solved with linear per-iteration costs \cite{li2018parallel, Ryu2017Unbalanced}, making them faster than even convolutional Sinkhorn approaches that have super-linear costs \cite{solomon2015convolutional}.

The second major drawback of the traditional OT formulation is a modeling deficiency: it restricts applications to the space of \emph{balanced} (i.e., fixed) mass systems (e.g., histograms and incompressible fluids), due to the assumption of conservation of mass.
When applied to images and video, this constraint is restrictive as intensity invariably changes (e.g., when non-spherical objects rotate or when luminescence varies).
To tackle this limitation, one may employ strategies such as \emph{partial} transport \cite{figalli2010optimal, caffarelli2010free}, which transports only a limited amount of mass between the systems, or \emph{unbalanced} transport \cite{chizat2016scaling, liero2018optimal}, which additionally models statistical properties of mass growth and decay to account for mass differences.

In this paper, we simultaneously tackle these two specific drawbacks to extend the applicability of optimal transport regularization to large-scale inverse imaging problems.
Our first contribution is a novel Beckmann formulation of the recent \emph{unbalanced} transport model \cite{chizat2016scaling}, which has more descriptive statistics than its partial transport counterpart.
This not only allows us to enrich OT's modeling capabilities, it also significantly reduces the number of variables from $\bigO{N^2}$ to $\bigO{N}$ making it tractable for large-scale imaging applications.
Our second contribution is a parallelizable and provably convergent iterative proximal algorithm for the proposed unbalanced-OT Beckmann formulation, facilitating compatibility with first-order solvers that are a modern staple in the image processing toolbox.
This framework lets us use the proposed unbalanced-OT as a variational loss in optimization problems, thereby introducing a new powerful tool to the variational OT \cite[\S9]{peyre2019computational} and computational inverse imaging domains.
To illustrate the utility of our proposed approach, we demonstrate superior reconstruction and support recovery performance in two computational inverse imaging applications involving dynamic tracking:
(i)  an online pixel-tracking estimation algorithm similar to dynamical methods such as \cite{Asif2010DynamicUpdating, Charles2016DynamicFiltering, Charles2017emd, bertrand2018earth}, and
(ii) a novel dynamical tracking method that augments robust PCA \cite{candes2011robust} with our unbalanced-OT regularizer, along with a parallel ADMM solver for this problem.

\section{Background}

\subsection{Optimal Transport}

Optimal transport (OT) is a rich mathematical subject \cite{villani2003topics} which originated with Monge \cite{monge1781memoire}, and was further developed by Kantorovich \cite{kantorovich2006problem} into the commonly used \emph{Monge-Kantorovich} formulation.
In its discretized form, it is a linear optimization program that minimizes the ``effort'' required to transport mass between histograms $\v{p}$ and $\v{q}$.
Transportation costs are defined using a \emph{ground cost} matrix $\v{C}$, whose entries specify the cost associated with moving one unit of mass between support locations (i.e., $C_{ij} = c(i,j)$, $\forall i \in \text{supp}(\v{p})$, and $\forall j \in \text{supp}(\v{q})$).
In this paper, we shall assume (without loss of generality) that $p$ and $q$ are supported over the same domain to ease notation.
Formally, the Monge-Kantorovich OT is defined as
\begin{equation} \label{eq:OT}
  \W (\v{p},\v{q}) =
  \min_{\v{P} \geq 0} \enspace
  \ip{\v{P}}{\v{C}}
  \quad \text{s.t.} \quad
  \v{P}\ones = \v{p},
  \v{P}^\top \ones = \v{q} ,
\end{equation}
where the notion of ``effort'' is defined as the product of mass and ground cost, as captured in the objective.
The inner product here denotes the Frobenius inner product (i.e., $\ip{\v{P}}{\v{C}} = \sum_{ij} P_{ij}C_{ij}$).
The matrix $\v{P} \in \reals^{N \times N}_+$ is a doubly stochastic matrix called the \emph{transport coupling}, whose entries $P_{ij}$ denote the amount of mass transported between locations $i$ and $j$ and whose marginals therefore equate to the given input histograms (per its constraints, where $\ones$ represents a vector of ones).
Let the $i,j$-th coordinate in $\v{p},\v{q}$'s support be denoted by $\delta_i,\delta_j$, respectively.
When $c(i,j) = (d(\delta_i,\delta_j))^p$ for a valid distance function $d(\cdot,\cdot)$, \eqref{eq:OT} is also known as the $p$-Wasserstein distance raised to the $p$-th power which we denote as $\W_p^p(\v{p},\v{q})$.

Although this definition of optimal transport has a meaningful geometric interpretation, it rules out a wide class of unbalanced signals (i.e., $\|\v{p}\|_1 \neq \|\v{q}\|_1$) such as images.
One could apply a na\"{\i}ve normalization strategy to ``force'' signals into the simplex but this comes at a risk of violating potentially important phenomenological considerations.
For example, consider a radar scenario where a target disappears from one frame to the next: normalization will arbitrarily and artificially increase energy of other targets to account for the one that has disappeared.
One well-known strategy to address this issue is called \emph{partial} transport \cite{figalli2010optimal, rubner2000earth}, which limits its transportation budget to only a fraction of mass in its arguments (i.e., $\ones^\top\v{P}\ones \leq \min ( \ones^\top\v{p} , \ones^\top\v{q} )$).
To convexify this constraint with respect to $\v{p}$ and $\v{q}$, a reformulation was proposed \cite{bertrand2018earth} to control the total transport budget:
\begin{equation}
\begin{aligned} \label{eq:Partial_OT}
  \U_{\mu} ( \v{p} , \v{q} )
  = \min_{\v{P} \geq 0} \enspace
  & \ip{\v{P}}{\v{C}} - \mu \ones^\top \v{P} \ones\\
  \text{s.t} \enspace
  & \v{P}\ones \leq \v{p}, \enspace
    \v{P}^\top\ones \leq \v{q}, \\
  & \ones^\top\v{P}\ones \leq \ones^\top\v{p} , \enspace
    \ones^\top\v{P}\ones \leq \ones^\top\v{q} ,
\end{aligned}
\end{equation}
where the objective's second term regulates the transport budget according to parameter $\mu > 0$.
Another recent strategy~\cite{chizat2016scaling, liero2018optimal} additionally models statistical properties of unaccounted mass via a mechanism of growth and decay.
We define \emph{unbalanced} transport as
\begin{equation} \label{eq:Unbal_OT}
  \V_\mu (\v{p},\v{q})
  = \min_{\v{P} \geq 0}
    \ip{\v{P}}{\v{C}}
  + \mu ( \| \v{P}\ones - \v{p} \|_p^p + \| \v{P}^\top\ones - \v{q} \|_p^p ) .
\end{equation}
The terms weighed by parameter $\mu > 0$ penalize unaccounted mass between the marginals of the optimal transport coupling and the input arguments, which analogously regulates the transport budget.
Although \eqref{eq:Partial_OT} and \eqref{eq:Unbal_OT} both possess mechanism to regulate the transportation budget, the $\ell_p$-norms in \eqref{eq:Unbal_OT} explicitly model growth/decay statistics while \eqref{eq:Partial_OT} models only the gross growth/decay budget.
For this reason and the fact that \eqref{eq:Unbal_OT} has a much simpler optimization (much fewer constraints) than \eqref{eq:Partial_OT}, we focus here on an unbalanced transport strategy.

Optimal transport is notoriously expensive to compute since each evaluation is itself a linear optimization with $\bigO{N^2}$ variables (i.e., size of $\v{P}$).
Fortunately, Beckmann provided a reformulation of the optimal transport problem using an incompressible fluid interpretation \cite{beckmann1952continuous}.
A key assumption with Beckmann's reformulation is that ground costs are Euclidean (i.e., $d(\delta_i,\delta_j) = \| \delta_i - \delta_j \|_2$), corresponding to the 1-Wasserstein distance.
Under this geodesic, mass cannot teleport and must follow straight-line paths between sources and sinks.
As a result, mass transport can be modelled using a flux field, which dramatically reduces the number of optimization variables from $\bigO{N^2}$ to $\bigO{N}$.
Beckmann's discrete optimal transport formulation may be stated as
\begin{equation} \label{eq:OT_beckman}
  \Wbeck ( \v{p} , \v{q} ) = \min_{\v{M}} \enspace \| \v{M} \|_{2,1} \enspace
  \text{s.t.} \enspace \divergence (\v{M}) - \v{q} + \v{p} = 0 ,
\end{equation}
where $\v{M} \in \reals^{N \times D}$ denotes a (fluidic) flux field, with $D$ representing the number of spatial dimensions of the field.
For this paper, we consider (without loss of generality) applications in imaging, so we let $D = 2$ and reorganize the columns of $\v{M}$ into $\v{M}_x , \v{M}_y \in \reals^{n_x \times n_y}$ representing the flux fields travelling in each dimension (with $N = n_x n_y$).
The objective therefore penalizes the presence of flux as:
\begin{gather*}
    \| \v{M} \|_{2,1} 
    :=
    \sum_{k=1}^N \| \v{M}_{k,:} \|_2 
    \equiv
    \sum_{i=1}^{n_x} \sum_{j=1}^{n_y} \sqrt{ \v{M}_x^2[i,j] + \v{M}_y^2[i,j]} ,
\end{gather*}
where indexing notation between column-major vector subscript $k \in \{ 1,\dots,N \}$ and pixel-wise coordinate indices $i \in \{1,\dots,n_x\}$ and $j \in \{1,\dots,n_y\}$ using external brackets are interchanged according to context.
In the constraint, the divergence of $\v{M}$, notated as $\divergence(\v{M})$, measures how much a discrete point in the flux field is a source or a sink:
\begin{gather*}
  \divergence(\v{M})_{k} \equiv \divergence(\v{M}) [i,j] \\
  := ( \v{M}_x[i,j] - \v{M}_x[i-1,j] ) + ( \v{M}_y[i,j] - \v{M}_y[i,j-1] ) ,
\end{gather*}
with zero-flux boundary conditions (i.e., $\v{M}_x [i,j] = \v{M}_y [i,j] = 0$ if $i$ or $j$ lies outside the support).
The constraint therefore balances all sources and sinks.
We note that this formulation has recently been successfully applied in \cite{li2018parallel} for computing the OT cost between large scale images and in \cite{bertrand2018earth} as a dynamical tracking regularizer to estimate streaming measurements in online fashion.

\subsection{Proximal first order methods}

Proximal first order optimization techniques often possess efficient per-iteration costs (typically parallelizable), along with reasonable convergence rates (typically at least linear).
For this reason, they have become an indispensable tool in large scale applications such as inverse imaging \cite{Figueiredo2007GradientProjection, afonso2010fast}.
Some popular proximal first order approaches include alternating direction method of multipliers (ADMM) \cite{boyd2011distributed}, the Douglas-Rachford algorithm \cite{lions1979splitting}, and primal-dual methods \cite{chambolle2011first}.
Variable splitting \cite{lions1979splitting, eckstein1989splitting} and the proximal point algorithm \cite{rockafellar1976monotone} are fundamental building blocks that underlie proximal first-order methods, where the idea is to divide a problem, $f(\v{x}) = \sum_i f_i(\v{x}_i)$, into a series of easier subproblems that iteratively converge to the global fixed point.
Each of these subproblems $f_i$ are solved with \textit{proximal operators}, which are defined as
\begin{equation*}
  \proximal_{\rho f_i} ( \v{z} ) = \argmin{\v{z}'} \Big\{ f_i(\v{z}') + \frac{1}{2\rho} \norm{\v{z} - \v{z}'}_2^2 \Big\} ,
\end{equation*}
where $f_i$ is any proper, convex, closed function.
The quadratic $\ell_2$ term (weighed by $\rho > 0$) serves as a strongly convex regularizer that ties each of the subproblems to the original problem, while simultaneously smoothing the subproblem; $\rho$ can be interpreted as a scalar step size.
Due to the highly iterative nature of first order methods, it is of paramount importance for the proximal algorithms of each subproblem to be efficient (e.g., have closed form solutions, or/and have a separable form that allows it to be solved with parallel hardware such as \emph{general purpose graphic processing units} -- GPGPUs).
We refer the reader to \cite{parikh2014proximal} for an excellent introductory monograph to the topic.



\section{Unbalanced Optimal Transport Regularizer}

In this section, we describe our novel Beckmann formulation of the unbalanced-OT, along with an accompanying iterative proximal algorithm for efficient implementation with proximal first order methods.

\subsection{Unbalanced Optimal Transport in Beckmann's Formulation}

We present a novel \emph{unbalanced}-OT (UOT) program based upon \eqref{eq:Unbal_OT}, presented in the style of Beckmann's formulation.
This proposed approach combines the best of two worlds: (i) it models important statistical aspects about the unaccounted mass between unbalanced signals, and (ii) the efficiency of the Beckmann formulation allows scalability of our method to large datasets.
We define the Beckmann formulation for unbalanced optimal transport as
\begin{equation} \label{eq:Unbal_OT_beckman}
\begin{aligned}
  \Vbeck_\mu ( \v{p} , \v{q} ) = \min_{\v{M},\v{r}}
  &\enspace \| \v{M} \|_{2,1} + \mu \| \v{r} \|_p^p \\
  \text{s.t.}
  &\enspace \divergence (\v{M}) - \v{q} + \v{p} = \v{r} .
\end{aligned}
\end{equation}
The key difference between this formulation and \eqref{eq:OT_beckman} is that there now exists a \emph{transport residual} term $\v{r} \in \reals^N$ that reflects the amount of mass that needs to be created or destroyed (i.e., unaccounted flux divergence) to balance the equality constraint.
To prevent unabated growth and decay of $\v{r}$, we penalize its magnitude with an $\ell_p^p$ norm, where $p$ is chosen based on statistics of the signal's support.
For example, if $\v{p},\v{q}$ are assumed sparse (i.e., $\norm{\v{p}}_0 \approx \norm{\v{q}}_0 \ll N$), we select $p=1$ to reflect a kurtotic distribution over the transport residual's support, meaning that growth/decay can only occur at a sparse number of locations.
Transportation cost is balanced against the cost of growth/decay using parameter $\mu > 0$.
We remark that in the same way that \eqref{eq:Unbal_OT} is a generalization of \eqref{eq:OT}, \eqref{eq:Unbal_OT_beckman} is also a generalization of \eqref{eq:OT_beckman}: a large $\mu$ drives the transport residual to a small value, therefore $\V_{\mu\to\infty} \to \W$. We also note that compared to partial transport \cite{Bertrand2018emd}, its single equality constraint makes it significantly easier to solve numerically compared to multiple inequality constraints.

\subsection{Proximal Operator of Unbalanced Optimal Transport}
\label{ssec:UOTprox}

Due to the critical role that proximal operators play in enabling first order solvers, we turn our attention to deriving this quantity for $\Vbeck_\mu$. This derivation will lead to the proposal of an associated efficient numerical solver and associated convergence guarantees.

The proximal operator $\proximal_{\rho\Vbeck}^{\mu} : (\reals^N \times \reals^N) \mapsto (\reals^N \times \reals^N)$ of $\Vbeck_\mu$ is defined as
\begin{equation} \label{eq:Prox_UOT}
\begin{aligned}
  &\proximal_{\rho\Vbeck}^{\mu} ( \v{p}_0 , \v{p}_1 ) \\
  &= \argmin{\v{x}_0,\v{x}_1 \geq 0}
  \Vbeck_\mu ( \v{x}_0 , \v{x}_1 ) + \frac{1}{2\rho} \norm{\begin{bmatrix}\v{x}_0\\\v{x}_1\end{bmatrix} - \begin{bmatrix}\v{p}_0\\\v{p}_1\end{bmatrix}}_2^2 \\
  &= \argmin{\v{x}_0,\v{x}_1 \geq 0, \v{M}, \v{r}}
  \| \v{M} \|_{2,1} + \mu \| \v{r} \|_p^p + \frac{1}{2\rho} \norm{\begin{bmatrix}\v{x}_0\\\v{x}_1\end{bmatrix} - \begin{bmatrix}\v{p}_0\\\v{p}_1\end{bmatrix}}_2^2 \\
  &\qquad\quad \text{s.t.} \enspace \divergence (\v{M}) - \v{x}_0 + \v{x}_1 = \v{r} .
\end{aligned}
\end{equation}
This objective is strongly convex and not everywhere infinite, so according to proximal operator theory it has a unique minimizer for every $(\v{p}_0,\v{p}_1)$. Recent work~\cite{li2018parallel} demonstrated that Chambolle-Pock's first order primal-dual method \cite{chambolle2011first} efficiently evaluates the \emph{balanced} OT problem \eqref{eq:OT_beckman}.
We extend this result and develop an efficient iterative algorithm to compute the \emph{proximal operator} of the \emph{unbalanced} OT problem \eqref{eq:Prox_UOT}.
This work, however, departs from \cite{li2018parallel} in two ways.
First, our work allows the arguments of the (unbalanced) OT problem to themselves become optimization variables, thus expanding the applicability of Beckmann OT to the domain of inverse imaging problems.
Second, we propose to exploit inexact and warm-start strategies that are crucial for scalability in inverse applications, allowing 2-3 orders of magnitude fewer iterations vis-\`a-vis the standalone Beckmann problem (i.e., $<10$ compared to $10^3$).
We note that in the final preparation of this paper, preliminary work describing a similar unbalanced-OT formulation in isolation (i.e., not as part of the general proximal framework we propose here) was concurrently developed~\cite{gangbo2019unnormalized}.
The Lagrangian of \eqref{eq:Prox_UOT} is:
\begin{equation}
\begin{aligned}
  \mathcal{L} (\v{M},\v{r},\v{x},\v{a}) = \enspace
  &\| \v{M} \|_{2,1} + \mu \| \v{r} \|_p^p + \frac{1}{2\rho} \norm{ \v{x}-\v{p} }_2^2 \\
  &+ \ind{+}(\v{x}) + \ip{\v{a}}{\divergence (\v{M}) + \v{A}\v{x} - \v{r}} ,
\end{aligned}
\end{equation}
where $\v{x} = \begin{bmatrix}\v{x}_0\\\v{x}_1\end{bmatrix}, \v{p} = \begin{bmatrix}\v{p}_0\\\v{p}_1\end{bmatrix}$, $\v{A} = [-\v{I},\v{I}]$, $\ind{+}$ is an indicator function of the non-negative orthant, and $\v{a}\in\reals^N$ is a Lagrange multiplier.
The saddle point to
\begin{equation}
\label{eqn:saddle}
  \min_{\v{M},\v{x},\v{r}} 
  \max_{\v{a}} \enspace
  \mathcal{L} (\v{M},\v{r},\v{x},\v{a})
\end{equation}
solves \eqref{eq:Prox_UOT}, for which the primal-dual method of Chambolle and Pock generates the following convergent sequence:
\begin{equation}
\begin{aligned} \label{eq:Prox_UOT_M}
  \v{M}\iternext \leftarrow \argmin{\v{M}}
  &\| \v{M} \|_{2,1} + \ip{\v{a}\itercur}{\divergence (\v{M})} \\
  &+ \frac{1}{2\tau_1} \norm{\v{M}-\v{M}\itercur}_F^2
\end{aligned}
\end{equation}
\begin{equation}
\begin{aligned} \label{eq:Prox_UOT_x}
  \v{x}\iternext \leftarrow \argmin{\v{x} \geq 0}
  &\frac{1}{2\rho}\| \v{x}-\v{p} \|_2^2 + \ip{\v{a}\itercur}{\v{A}\v{x}} \\
  &+ \frac{1}{2\tau_1} \norm{\v{x}-\v{x}\itercur}_2^2
\end{aligned}
\end{equation}
\begin{equation}
\begin{aligned} \label{eq:Prox_UOT_r}
  \v{r}\iternext \leftarrow \argmin{\v{r}}
  &\mu \| \v{r} \|_p^p + \ip{\v{a}\itercur}{-\v{r}} \\
  &+ \frac{1}{2\tau_1} \norm{\v{r}-\v{r}\itercur}_2^2
\end{aligned}
\end{equation}
\begin{equation} \label{eq:Prox_UOT_a}
  \v{a}\iternext \leftarrow \argmax{\v{a}}
  \ip{\v{a}}{\v{b}\iternext} - \frac{1}{2\tau_2} \norm{\v{a}-\v{a}\itercur}_2^2 ,
\end{equation}
where $\v{b}\iternext = 2 \mathcal{K}(\v{M}\iternext,\v{x}\iternext,\v{r}\iternext) - \mathcal{K}(\v{M}\itercur,\v{x}\itercur,\v{r}\itercur)$ and $\mathcal{K}(\v{M},\v{x},\v{r}) = \divergence (\v{M}) + \v{A}\v{x} - \v{r}$.

Updates \eqref{eq:Prox_UOT_M}-\eqref{eq:Prox_UOT_a} all have standard closed-form proximal algorithms \cite{parikh2014proximal}.
For \eqref{eq:Prox_UOT_M}, we exploit the fact that the program is row-wise separable.
Denoting each row of $\v{M}\iternext$ as $\v{m}\iternext_i \in \mathbb{R}^2, \forall i=\{1,\dots,N\}$,  we apply the $\ell_{2,1}$-norm proximal algorithm (vector-soft shrinkage operator) on each of its rows:
\begin{equation} \label{eq:Prox_UOT_M_prox}
  \v{m}\iternext_i = \shrink^{\ell_2}_{\tau_1} ( \v{m}\itercur_i - \tau_1 \divergence^\ast(\v{a}\itercur)_i ) ,
\end{equation}
where $\divergence^\ast(\cdot) : \reals^{N} \mapsto \reals^{N \times 2}$ refers to the adjoint of the $\divergence(\cdot)$ operator and where $\shrink^{\ell_2}_\sigma (\v{q}) = \max \{ \norm{\v{q}}_2 - \sigma , 0 \} \odot \frac{\v{q}}{\norm{\v{q}}_2}$.
Next, \eqref{eq:Prox_UOT_x} applies a projection unto the non-negative orthant:
\begin{equation}
  \v{x}\iternext = \Pi_+\Big( \frac{\rho\tau_1}{1+\rho\tau_1}\v{p} + \frac{1}{1+\rho\tau_1} \big(\v{x}\itercur - \tau_1 \v{A}^\top\v{a}\itercur \big) \Big) ,
\end{equation}
where $\Pi_+(\v{q}) = \max \{ \v{q} , 0\}$.
For \eqref{eq:Prox_UOT_r}, we let $p=1$ here to apply a linear penalty on mass growth/decay.
Applying an $\ell_1$ shrinkage operator yields:
\begin{equation}
  \v{r}\iternext = \shrink^{\ell_1}_{\mu\tau_1} ( \v{r}\itercur + \tau_1 \v{a}\itercur ),
\end{equation}
where $\shrink^{\ell_1}_\sigma (\v{q}) = \sign(\v{q}) \odot \max \{ \abs{\v{q}} - \sigma , 0 \}$.
We remark that the update step for $\v{r}$ with $p=2$ may be alternatively derived to be an averaging update.
Lastly, \eqref{eq:Prox_UOT_a} is a simple projection:
\begin{equation}
  \v{a}\iternext = \v{a}\itercur + \tau_2 \v{b}\iternext .
\end{equation}
The full algorithm is presented in Algorithm \ref{algo:UOTProx}.

\begin{algorithm}
\caption{Unbalanced Beckmann OT Proximal Algorithm.}
\label{algo:UOTProx}
\begin{algorithmic}[1]
\Require $\v{p}$, $\v{M}^{0}$, $\v{x}^{0}$, $\v{r}^{0}$, $\v{a}^{0}$, $\mu$, $\rho$, $\tau_1$, $\tau_2$
\Ensure $\v{M}\iternext,\v{x}\iternext,\v{r}\iternext,\v{a}\iternext$
\State $k = 1$
\While {\text{not converged}}
    \State $\v{m}\iternext_i \leftarrow \shrink^{\ell_2}_{\tau_1} ( \v{m}\itercur_i - \tau_1 \divergence^\ast(\v{a}\itercur)_i ), \enspace \forall i$
    \State $\v{x}\iternext \leftarrow \Pi_+ \big( \frac{\rho\tau_1}{1+\rho\tau_1}\v{p} + \frac{1}{1+\rho\tau_1} (\v{x}\itercur - \tau_1 \v{A}^\top\v{a}\itercur ) \big)$
    \State $\v{r}\iternext \leftarrow \shrink^{\ell_1}_{\mu\tau_1} ( \v{r}\itercur + \tau_1 \v{a}\itercur )$
    \State $\v{b}\iternext \leftarrow 2 \mathcal{K}(\v{M}\iternext,\v{x}\iternext,\v{r}\iternext) - \mathcal{K}(\v{M}\itercur,\v{x}\itercur,\v{r}\itercur)$
    \State $\v{a}\iternext \leftarrow \v{a}\itercur + \tau_2 \v{b}\iternext$
    \State $k \leftarrow k + 1$
\EndWhile
\end{algorithmic}
\end{algorithm}

Although an explicit parallel implementation (e.g., GPU) is beyond the scope of this paper, we explain key parallel-implementation considerations of this algorithm.
Similar concepts have been presented in \cite{li2018parallel} but are included here for completeness.
The $\divergence$ operator (applied in lines 3 and 6) is a sparse operator that couples pixels, necessitating points of synchronization at each proximal iteration.
For example, after $\divergence^*(\v{a}\itercur)$ is synchronously computed in line 3, lines 3 thru 5 (which are themselves element-wise separable) can be computed in parallel. The memory complexity of this algorithm is also linear, since the memory-bottleneck (the $\divergence$ operator) has an implicit implementation.
These properties make this algorithm amenable for parallelization.

Next, we discuss two aspects of run time that are important in practical applications: the per-iteration computational cost and the required number of iterations.
The computational bottleneck is due to the $\divergence$ operator which can be computed in linear time since it is a sparse matrix containing $\leq 4$ nonzero entries per row.  Thus, the per-iteration cost of this algorithm is linear.
Algorithm~\ref{algo:UOTProx} is intended to be used as a part of a larger solver (e.g., Algorithm~\ref{algo:UOTDF} or Algorithm~\ref{algo:RPCA_UOTDF}), where its proximal iterations are inner loops within outer loops of the larger solver.
Since a large number of inner loops is undesirable, we propose to apply warm-starts and early-termination strategies. We have found these approaches to significantly reduce the required number of iterations (demonstrated in Section \ref{sssec:prox_efficiency}), allowing its practical run-time to be linear (verified in Section~\ref{sssec:UOTDF_ADMM}).

We also describe the special case of \eqref{eq:Prox_UOT} when one of its arguments are constant. Let us define the proximal operator $\proximal_{\rho\Vbeck_{\v{s}}}^\mu : \reals^N \mapsto \reals^N$ of $\Vbeck_\mu$ with a constant $\v{s}\in\reals^N$ in the first argument as
\begin{equation}
    \label{eq:Prox_UOT_specialcase}
    \proximal_{\rho\Vbeck_{\v{s}}}^\mu (\v{p}) =
    \argmin{\v{x}} \Vbeck_\mu(\v{s},\v{x}) + \frac{1}{2\rho} \norm{\v{x}-\v{p}}_2^2 .
\end{equation}
Since this is solved in a similar fashion as above, Algorithm~\ref{algo:UOTProx} is applied with only with a difference in $\v{A}$ and $\mathcal{K}$:
\begin{equation*}
\begin{aligned}
    \v{A} := \v{I} , \quad\quad
    \mathcal{K}(\v{M},\v{x},\v{r}) := \divergence(\v{M}) - \v{s} + \v{x} - \v{r} . 
\end{aligned}
\end{equation*}

We conclude this section with an analytic guarantee that  specifies step size conditions for the convergence of Algorithm~\ref{algo:UOTProx} to the saddle point of \eqref{eqn:saddle}.
\begin{theorem}[Convergence guarantee]
\label{thm}
Suppose $\tau_1 \tau_2 < \frac{1}{\lambda_{\max}(\nabla^2) + 3}$ where $\lambda_{\max}(\nabla^2)$ is the largest eigenvalue of the discrete Laplacian operator, $\nabla^2$.
Then the steps in Algorithm~1 produce a series which converges to the saddle point of the Lagrangian, i.e.,
\[
  (\v{M}^k, \v{r}^k, \v{x}^k, \v{a}^k) \to (\v{M}^\star, \v{r}^\star, \v{x}^\star, \v{a}^\star),
\]
where $(\v{M}^\star, \v{r}^\star, \v{x}^\star, \v{a}^\star)$ is a solution of \eqref{eqn:saddle}.
\end{theorem}

\begin{proof}
We proceed by showing that the conditions of Theorem~1 in \cite{chambolle2011first} are satisfied.
First, note that we may rewrite \eqref{eq:Prox_UOT} as
\[
\mathcal{L}(\v{M}, \v{r}, \v{x}, \v{a}) = G(\v{M}, \v{r}, \v{x}) + \ip{\v{a}}{\v{K}\v{b}} - F(\v{a}),
\]
where $G(\v{M}, \v{r}, \v{x}) = \norm{\v{M}}_{2,1} + \mu\norm{\v{r}}_p^p + \frac{1}{2\rho}\norm{\v{x}-\v{p}}_2^2 + \ind{+}(\v{x})$, $\v{K} = [\div, -\v{I}, \v{I}, -\v{I}]$, $\div$ is the matrix corresponding to the divergence operator, $\v{b} = [\vec(\v{M}); \v{x}; \v{r}]$ and $F(\v{a}) = 0$.
The functions $G$ and $F$ are proper, convex, and lower semi-continuous, and $K$ is a linear operator.
Note that if $\lambda$ is an eigenvalue of a matrix $\v{B}^* \v{B}$, then $\lambda + 1$ is an eigenvalue of the matrix $\alt{B}^* \alt{B}$ where $\alt{B} = [\v{B}, \pm \v{I}]$ (this is easily verified by noting that if $\v{v}$ is the corresponding eigenvector of $\v{B}^* \v{B}$, then $[\lambda \v{v}; \pm\v{Bv}]$ is the corresponding eigenvector of $\alt{B}^*\alt{B}$).
By repeated application of this identity, it follows that $\lambda_{\max}(\v{K}) = \lambda_{\max}(\div^* \div) + 3$.
Since $\div \div^*$ is the discrete Laplacian operator, $\lambda_{\max}(\nabla^2) + 3$ is the maximum eigenvalue of $\v{K}$.
Thus, the conditions of \cite[Theorem 1]{chambolle2011first} are satisfied when $\tau_1 \tau_2 < \frac{1}{\lambda_{\max}(\nabla^2) + 3}$.
\end{proof}


\section{Applications in Dynamical Tracking}

In this section, we demonstrate our unbalanced-OT model as a regularizer in dynamical tracking applications \cite{Charles2016DynamicFiltering, Charles2017emd, Bertrand2018emd, bertrand2018earth} that utilize novel models to encourage continuity between snapshots of time-varying data.
We represent snapshots of target activations with the state vector $\v{s}_t \in \reals^N$ where $t$ is the time index.
Targets are assumed mobile and $K$-sparse over a gridded support where $\v{s}_t$ evolves dynamically in time and space via a first-order Markovian relationship through some function $f_t$
\begin{equation*}
  \v{s}_t = f_t ( \v{s}_{t-1} ) + \v{\nu}_t ,
\end{equation*}
where $\v{\nu}_t \in \reals^{N}$ captures dynamical innovations (i.e., model mismatch).

We apply our unbalanced-OT regularizer in two dynamical tracking settings:
(i) causal tracking with streaming measurements in Section~\ref{ssec:BPDN_online_setting} and
(ii) robust principal components (RPCA) \cite{candes2011robust} for batch settings in Section~\ref{ssec:RPCA_batch_setting}.




\subsection{Online dynamical estimation with UOT-regularization}
\label{ssec:BPDN_online_setting}

For our first application, we consider the problem of dynamically tracking $\v{s}_t$ from streaming measurements
\begin{equation*}
    \v{y}_t = \v{\Phi}_t \v{s}_t + \v{\eta}_t ,
\end{equation*}
where $\v{\eta}_t \in \reals^M$ represents an iid Gaussian noise model.
The goal in this tracking problem is to accurately recover $\v{s}_t$ at each time step given $\v{\Phi}_t$ and $\v{y}_t$. 

Classical dynamical tracking methods like the Kalman filter \cite{Kalman1960NewApproach} exploit temporal structure for estimation and inference.
Under a probabilistic \emph{maximum a posteriori} framework, the Kalman filter provides the optimal\footnote{Under Gaussian assumptions on the measurement and dynamical noise.} estimate of the current signal under a linear dynamical function (i.e., $f_t(\v{s}_{t-1}) = \v{F}_t \v{s}_{t-1}$).
Concisely stated, the Kalman filter solves
\begin{equation*}
  \hat{\v{s}}_t 
  =
  \argmin{\v{s}} \| \v{y}_t - \v{\Phi}_t \v{s} \| 
  + \| \v{s} - \v{F}_t \hat{\v{s}}_{t-1} \| ,
\end{equation*}
which balances an observation term and a prediction term (using norms that capture noise statistics).
The $\ell_p$ losses are only a sensible design choice  under a \emph{Lagrangian} state-space representation where each state variable $\v{s}_t[i]$ is itself moving through a geometric space (e.g., displacement coordinates of a GPS sensor).
A design flaw arises if the state space is \emph{Eulerian}-represented (i.e., the support space is gridded and signals move across its support).
When displacement-variations are expected, it makes sense to pay penalties that are proportional to support-displacement error, yet $\ell_p$-norm penalties are invariant to this, rendering $\ell_p$-norms ineffective at support estimation.
In spite of this, \emph{Eulerian}-represented signals are still of interest because they are extremely effective when signal-support statistics are known \emph{a priori} (e.g., sparsity statistics \cite{tibshirani1996regression, Charles2017emd}).
The unbalanced-OT model is a suitable regularizer that circumvents support-estimation issues, since it inherently accounts for the geometry of an \emph{Eulerian} representation's support.
We note that optimal transport priors have recently been proposed in inverse imaging problems \cite{karlsson2017generalized, abraham2017tomographic} as well as prior work in dynamical tracking \cite{Charles2017emd, Bertrand2018emd, bertrand2018earth}.

We propose a least-squares dynamic filtering formulation with unbalanced-OT regularization (UOT-DF) under assumptions of first-order Markovian dynamics:
\begin{equation} \label{eq:UOT-DF}
  \hat{\v{s}}_t =
  \argmin{\v{s} \geq 0}
    \frac{1}{2} \norm{\v{y}_t - \v{\Phi}_t \v{s}}_2^2
  + \kappa \Vbeck_\mu ( \v{s} , \tilde{\v{s}} ) ,
\end{equation}
where parameter $\kappa > 0$ balances between data fidelity (first term) and dynamics (second term). Here, $\tilde{\v{s}} = f_t(\hat{\v{s}}_{t-1}) \in \reals^{N}$ represents the prediction of the current state, formed by propagating the previous state forward in time via the dynamical process $f_t$.
The simplicity of this formulation helps us better understand the behavior of different types of OT-regularization; specifically, we characterize differences between balanced-OT and unbalanced-OT in Section~\ref{sssec:results_comparison_OT_schemes}.

We also propose a formulation based on basis pursuit de-noising (BPDN) \cite{Donoho2005SparseNonnegative}, where a sparsity regularizer is additionally included into the objective of \eqref{eq:UOT-DF}. We call this formulation the unbalanced-OT regularized BPDN (BPDN+UOT-DF)
\begin{equation}
    \label{eq:BPDN+UOT-DF}
    \hat{\v{s}}_{t} = \argmin{\v{s} \geq 0} 
    \frac{1}{2} \norm{\v{y}_t - \v{\Phi}_t \v{s} }_2^2
  + \lambda \norm{\v{s}}_1 
  + \kappa \Vbeck_{\mu} (\v{s},\tilde{\v{s}}) ,
\end{equation}
where parameters $\lambda,\kappa>0$ balance the ratio between the objective's terms.
This formulation is a natural progression from prior art in the sparse tracking literature \cite{Charles2016DynamicFiltering, Charles2017emd, Bertrand2018emd}, which we compare BPDN+UOT-DF against in Section~\ref{sssec:results_literature_comparison}.



\subsubsection{ADMM solver for UOT-DF}
\label{sssec:UOTDF_ADMM}

First, we describe in detail how our unbalanced-OT Beckmann proximal algorithm may be applied to efficiently solve problems \eqref{eq:UOT-DF} or \eqref{eq:BPDN+UOT-DF} within an ADMM framework \cite{boyd2011distributed}. 
Consider the general problem
\begin{equation} \label{eq:general_UOT-DF}
  \min_{\v{s} \geq 0}
    \frac{1}{2} \norm{\v{y} - \v{\Phi} \v{s}}_2^2
  + \mathcal{R}(\v{s})
  + \kappa \Vbeck_\mu ( \v{s} , {\v{s}}_{0} ) ,
\end{equation}
where $\mathcal{R}$ refers to a proper, convex, closed function that represents a prior (e.g., $\ell_1$-norm) describing structure in $\v{s}$ (e.g., sparsity).
We apply a variable splitting procedure to formulate an equivalent problem
\begin{equation*}
  \min_{\v{s} = \v{x} = \v{z}}
    \frac{1}{2} \norm{\v{y} - \v{\Phi} \v{x}}_2^2
  + \mathcal{R}(\v{s}) + \iota_+(\v{s})
  + \kappa \Vbeck_\mu ( \v{z} , {\v{s}}_{0} ) ,
\end{equation*}
where $\v{x},\v{z}$ are splitting variables.
Its \textit{augmented} Lagrangian is formulated as
\begin{equation*}
\begin{split}
    \mathcal{L}(\v{s},\v{x},\v{z},\v{a},\v{b}) = 
    \frac{1}{2} \norm{\v{y} - \v{\Phi} \v{x}}_2^2
  + \mathcal{R}(\v{s}) + \iota_+(\v{s}) \\
  + \kappa \Vbeck_\mu ( \v{z} , {\v{s}}_{0} )
  + \frac{\rho}{2} \Big(
    \norm{\v{x}-\v{s}+\v{a}}_2^2
  + \norm{\v{z}-\v{s}+\v{b}}_2^2
  \Big) ,
\end{split}
\end{equation*}
where $\v{a},\v{b}$ here are dual variables.
From this, ADMM generates a convergent sequence of iterative updates by successively solving each primal variable and taking a gradient step in the dual space.
Letting the iteration index be expressed with superscript $k$, we have the following updates.
\begin{equation*}
\begin{aligned}
  \v{s} \iternext
  & \leftarrow \argmin{\v{s}}
    \mathcal{R}(\v{s})
  + \iota_+(\v{s})
  + \frac{\rho}{2} \norm{\v{s} - (\v{x}\itercur + \v{a}\itercur) }_2^2 \\
  & = \proximal_{\mathcal{R}+}^{\rho} (\v{x}\itercur +\v{a}\itercur) .
\end{aligned}
\end{equation*}
In this work, we are primarily concerned with two cases
(i) UOT-DF \eqref{eq:UOT-DF} where $\mathcal{R}(\v{s}) = 0$ which results in the element-wise operator $\proximal_{\mathcal{R}+}^{\rho} (\v{q}) = \Pi_+(\v{q})$, and
(ii) BPDN+UOT-DF \eqref{eq:BPDN+UOT-DF} $\mathcal{R}(\v{s}) = \lambda \norm{\v{s}}_1$ which results in the element-wise operator $\proximal_{\mathcal{R}+}^{\rho} (\v{q}) = \Pi_+(\v{q}-\rho/\lambda)$.
\begin{equation*}
\begin{aligned}
  \v{x} \iternext
  & \leftarrow \argmin{\v{x}}
    \frac{1}{2} \norm{\v{y} - \v{\Phi} \v{x}}_2^2
  + \frac{\rho}{2} \norm{\v{x} - (\v{s}\iternext - \v{a}\itercur) }_2^2 \\
  & = (\v{\Phi}^\top\v{\Phi} + \rho\v{I})^{-1} ( \v{\Phi}^\top \v{y} + \rho (\v{s}\iternext - \v{a}\itercur) ).
\end{aligned}
\end{equation*}
\begin{equation*}
\begin{aligned}
  \v{z} \iternext
  & \leftarrow \argmin{\v{z}}
    \Vbeck_{\mu}(\v{z},\v{s}_0)
  + \frac{1}{2\kappa/\rho} \norm{\v{z} - (\v{s}\iternext - \v{b}\itercur) }_2^2 \\
  & = \proximal^\mu_{\kappa/\rho\Vbeck_{\v{s}_0}} (\v{s}\iternext - \v{b}\itercur) ,
\end{aligned}
\end{equation*}
where this proximal operator refers to the special case \eqref{eq:Prox_UOT_specialcase} which is solved with a slight modification to Algorithm~\ref{algo:UOTProx} as previously discussed:
\begin{equation*}
\begin{aligned}
  \v{a} \iternext &\leftarrow \v{a}\itercur + \v{x}\iternext - \v{s}\iternext, \\
  \v{b} \iternext &\leftarrow \v{b}\itercur + \v{z}\iternext - \v{s}\iternext. \\
\end{aligned}
\end{equation*}

The full algorithm is presented in Algorithm~\ref{algo:UOTDF}.
Line 6 requires an application of Algorithm~\ref{algo:UOTProx} whose output is a tuple containing several variables (i.e., $\v{z}$ as well as auxiliary/dual variables). In practice, these variables are initialized to $0$ and cached for subsequent warm-starts.
With the exception of lines 5 and 6, all lines within the while loop have linear run-time complexity.
The complexity of line 6 can be reduced to linear by exploiting warm-starts and early termination.
Therefore, the dominant cost is the matrix inversion step of line 5, which ranges between $\bigO{N}$ (when $\v{\Phi}$ is diagonal) and $\bigO{N^2}$ (when $\v{\Phi}$ is dense).

\begin{algorithm}
\caption{UOT-DF ADMM Algorithm.}
\label{algo:UOTDF}
\begin{algorithmic}[1]
\Require $\v{y}$, $\v{\Phi}$, $\v{s}_0$, $\kappa>0$, $\mu>0$, $\rho>0$
\Ensure $\v{s}\iternext$
\State $k = 1$
\State $\v{s}\itercur \leftarrow \v{x}\itercur \leftarrow \v{z}\itercur \leftarrow \v{a}\itercur \leftarrow \v{b}\itercur \leftarrow \v{0}$.
\While {\text{not converged}}
    \State $\v{s} \iternext \leftarrow \proximal_{\mathcal{R}+}^{\rho} (\v{x}\itercur +\v{a}\itercur)$
    \State $\v{x} \iternext \leftarrow (\v{\Phi}^\top\v{\Phi} + \rho\v{I})^{-1} ( \v{\Phi}^\top \v{y} + \rho (\v{s}\iternext - \v{a}\itercur) )$
    \State $\v{z} \iternext \leftarrow \proximal^\mu_{\kappa/\rho\Vbeck_{\v{s}_0}} (\v{s}\iternext - \v{b}\itercur)$
    \State $\v{a} \iternext \leftarrow \v{a}\itercur + \v{x}\iternext - \v{s}\iternext$
    \State $\v{b} \iternext \leftarrow \v{b}\itercur + \v{z}\iternext - \v{s}\iternext$
    \State $k \leftarrow k + 1$
\EndWhile
\end{algorithmic}
\end{algorithm}

We present results that showcase the scalability of our unbalanced Beckmann OT proximal algorithm in the applied context of Algorithm~\ref{algo:UOTDF}.
In Figure~\ref{fig:011_runtime_vs_problemsize}, we measure the time (wall time) required to solve \eqref{eq:UOT-DF} as a denoising problem (i.e., $\v{\Phi} = \v{I}$ and $\mathcal{R}=0$).
We generate synthetic problems across problem sizes $N = \{8^2,16^2,32^2,64^2,128^2,256^2,512^2\}$, where $N$ refers to the number of pixels.
To achieve algorithmic efficiency of the proximal algorithm in line 6, we employ warm-starts and terminate after a single iteration (refer to Section~\ref{sssec:prox_efficiency} for justification).
The stopping criteria of our ADMM algorithm is:
\begin{equation*}
    \norm{\begin{bmatrix}\v{x}\iternext-\v{s}\iternext\\\v{z}\iternext-\v{s}\iternext\end{bmatrix}}_2 < \varepsilon , \quad
    \rho \norm{\begin{bmatrix}\v{x}\iternext-\v{x}\itercur\\\v{z}\iternext-\v{z}\itercur\end{bmatrix}}_2 < \varepsilon ,
\end{equation*}
where the primal and dual residuals are respectively less than $\varepsilon = 10^{-3}$ (see \cite[\S 3.3.1]{boyd2011distributed}).
We report the number of iterations required for convergence and the per-iteration wall time, which validates the main computation advantage of our UOT proximal algorithm: its linear run-time complexity.
Median values are reported and the experiments were performed on an Intel 6-core i7 2.60 GHz CPU with 16 GB of RAM.

\begin{figure}[htb]
  \centering
  \centerline{\includegraphics[trim={0cm 0.4cm 0cm 0cm},clip,width=0.49\textwidth]
  {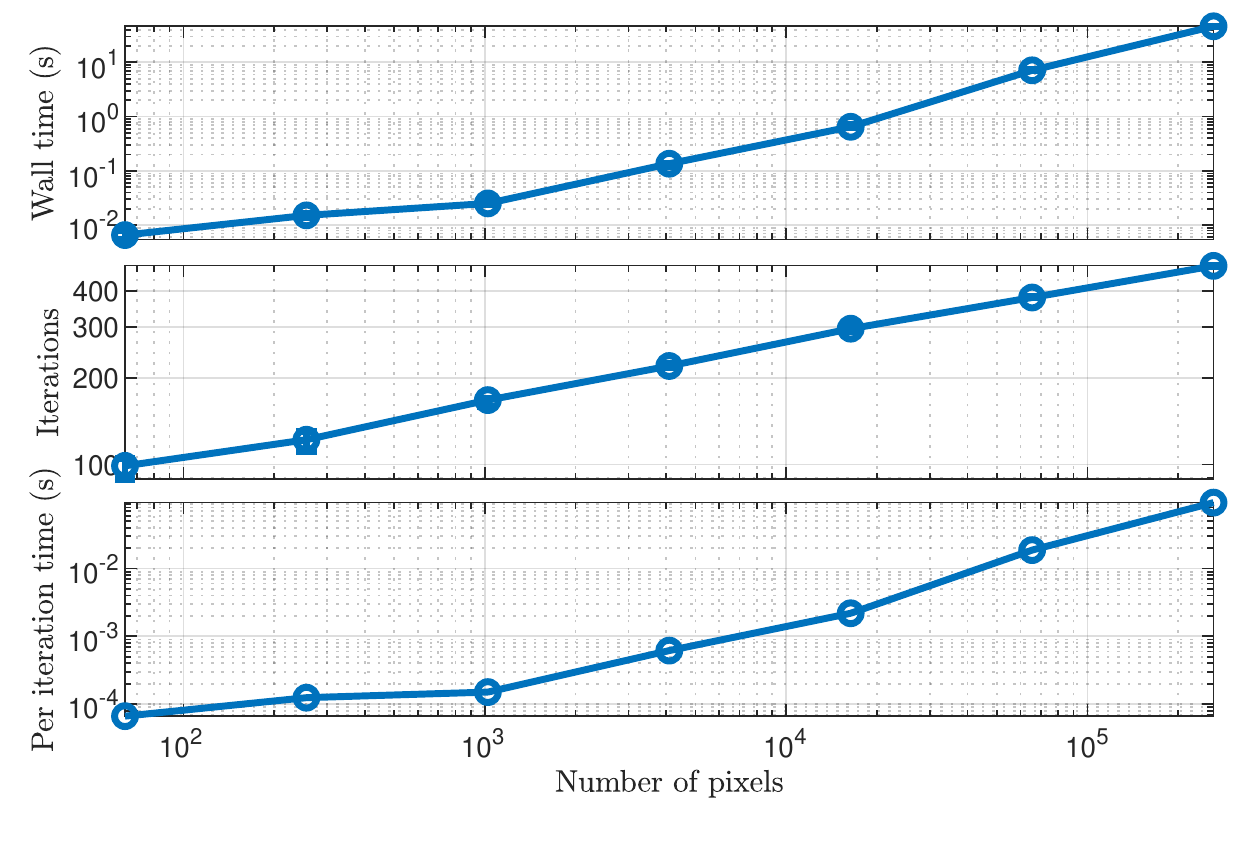}}
\caption{
  Runtime plots of Algorithm~\ref{algo:UOTDF} on different problem sizes.
  The per-iteration time reflects the linear run-time complexity of our algorithm.
}
\label{fig:011_runtime_vs_problemsize}
\end{figure}


\subsubsection{Balanced-OT versus unbalanced-OT regularization under mass-changing regimes}
\label{sssec:results_comparison_OT_schemes}

To motivate this problem, we consider a radar tracking application where targets are represented as single (pixel) activations that move about independently over a fixed grid, and where amplitude-mismatch occurs between frames in time.
OT-regularization depends on principles of mass-conservation, thus, inference is negatively affected when such assumptions are violated.
We conduct experiments where the total mass changes between video frames to better understand the behaviors of balanced- and unbalanced-OT.
We simplify our study by considering the problem of single-frame reconstruction with an OT-regularizer
\begin{equation}
  \min_{\v{s}}
  \frac{1}{2} \norm{\v{y} - \v{\Phi} \v{s}}_2^2 + \kappa \mathcal{T} ( \v{s} , \v{s}_0 ) ,
\end{equation}
where $\v{y}\in\reals^M$ is a noisy observation, $\v{\Phi}^{M \times N}$ is a random Gaussian matrix ($M/N = 0.35$), $\v{s}_0$ is a dynamical-prior (with an identity dynamical function, i.e., $f_t(\v{s}_0) = \v{s}_0$), and the OT-regularizer $\mathcal{T}$ here takes either OT strategies:
(i)   balanced-OT ($\mathcal{T} = \Wbeck$) according to \eqref{eq:OT} which assumes that the total mass is equal to that of the prior-frame, and
(ii) unbalanced-OT ($\mathcal{T} = \Vbeck_{\mu}$) according to \eqref{eq:Unbal_OT} which assumes total amount of transported mass is regulated by an $\ell_p^p$-norm penalty on growth/decay of mass, with $p=1$ in this experiment.

We set up an illustrative denoising simulation, where sparse $10 \times 10$-pixel images have a total mass that changes between frames.
We allow mass to change under two regimes: mass growth and mass decay.
The goal is to recover the second frame $\v{s}$ from a compressed noisy observation $\v{y} = \v{\Phi} \v{s} + \v{\eta}$, where $\v{\eta} \sim \mathcal{N} (0,\sigma^2\v{I})$.
To simplify our empirical analysis, we assign $\v{s}_0$ to be the (uncorrupted) first frame.
The rate of mass change is defined as $|\ones^\top\v{s} - \ones^\top\v{s}_0| / \ones^\top \v{s}$.
Under the growth regime the intensity of active pixels increase, while they decrease under the decay regime.
Between frames, spatial support movement of active pixels is randomly assigned according to a radial Gaussian probability whose mean is one pixel away from its original location.
Algorithm~\ref{algo:UOTDF} was applied to solve both methods\footnote{Balanced-OT requires the \textit{balanced} version of Algorithm~\ref{algo:UOTProx}, which is a trivial modification: variable $\v{r}$ is dropped, and $\mathcal{K}(\v{M},\v{x}) := \divergence(\v{M}) - \v{s} + \v{x}$.} with line 4 assuming $\proximal_{\mathcal{R}+}^\rho(\v{q}) = \Pi_+(\v{q})$.
An exhaustive parameter search is performed at each experimental setting to obtain the optimal performance of each method.


\begin{figure}[htb]
  \centering
  \centerline{\includegraphics[trim={0cm 0cm 0cm 0cm},clip,width=0.45\textwidth]
  {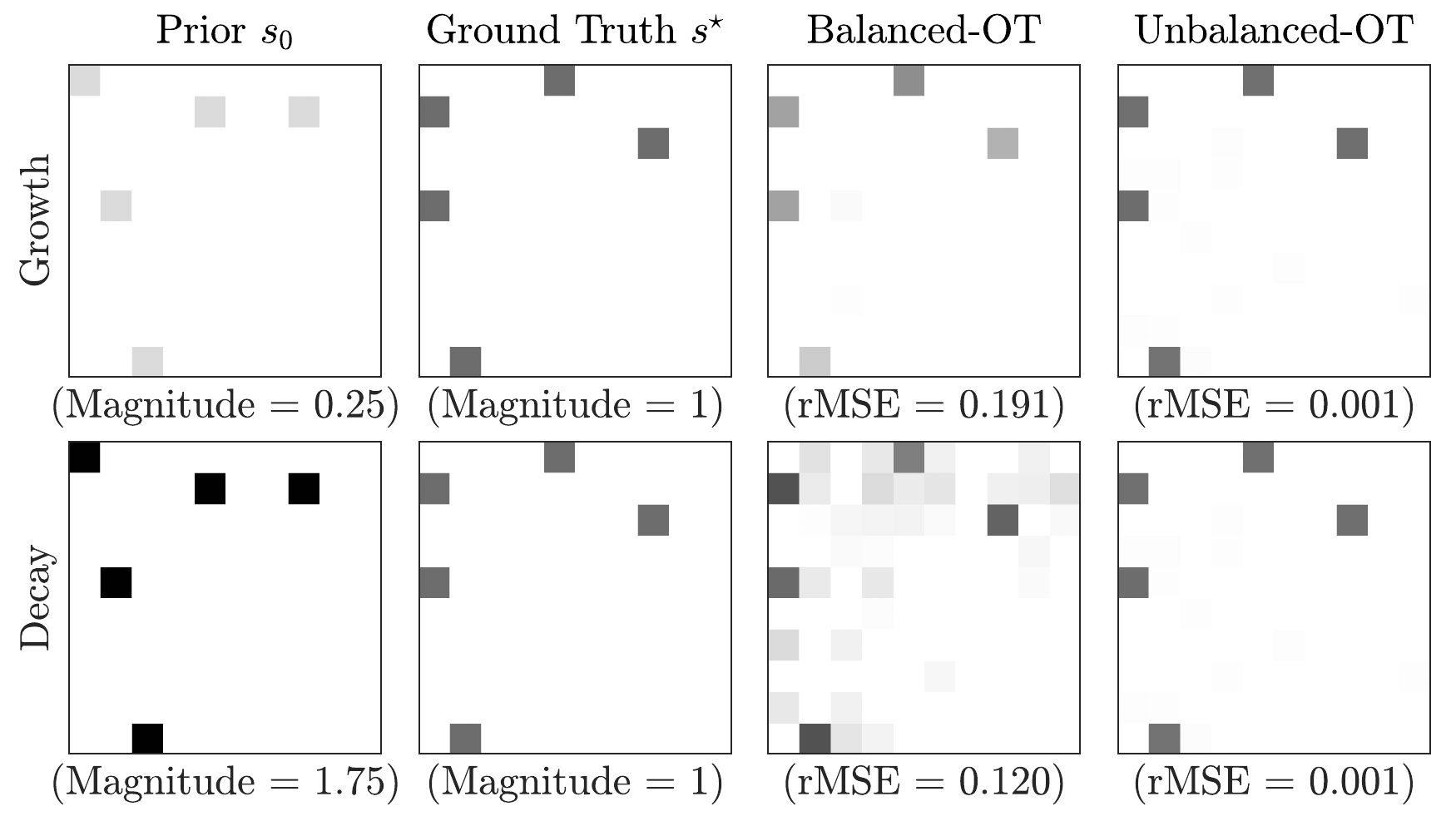}}
\caption{
  Qualitative reconstruction performance of OT-regularization schemes under mass changing regimes.
  We observe the optimal reconstruction solutions from denoising with various OT-priors under regimes of mass-growth (top row) and mass-decay (bottom row).
  Across both regimes, unbalanced-OT offers solutions with good support estimation and low relative mean squared error (rMSE).
}
\label{fig:007_qualitative}
\end{figure}

\begin{figure}[htb]
  \centering
  \centerline{\includegraphics[trim={0cm 0cm 0cm 0cm},clip,width=0.49\textwidth]
  {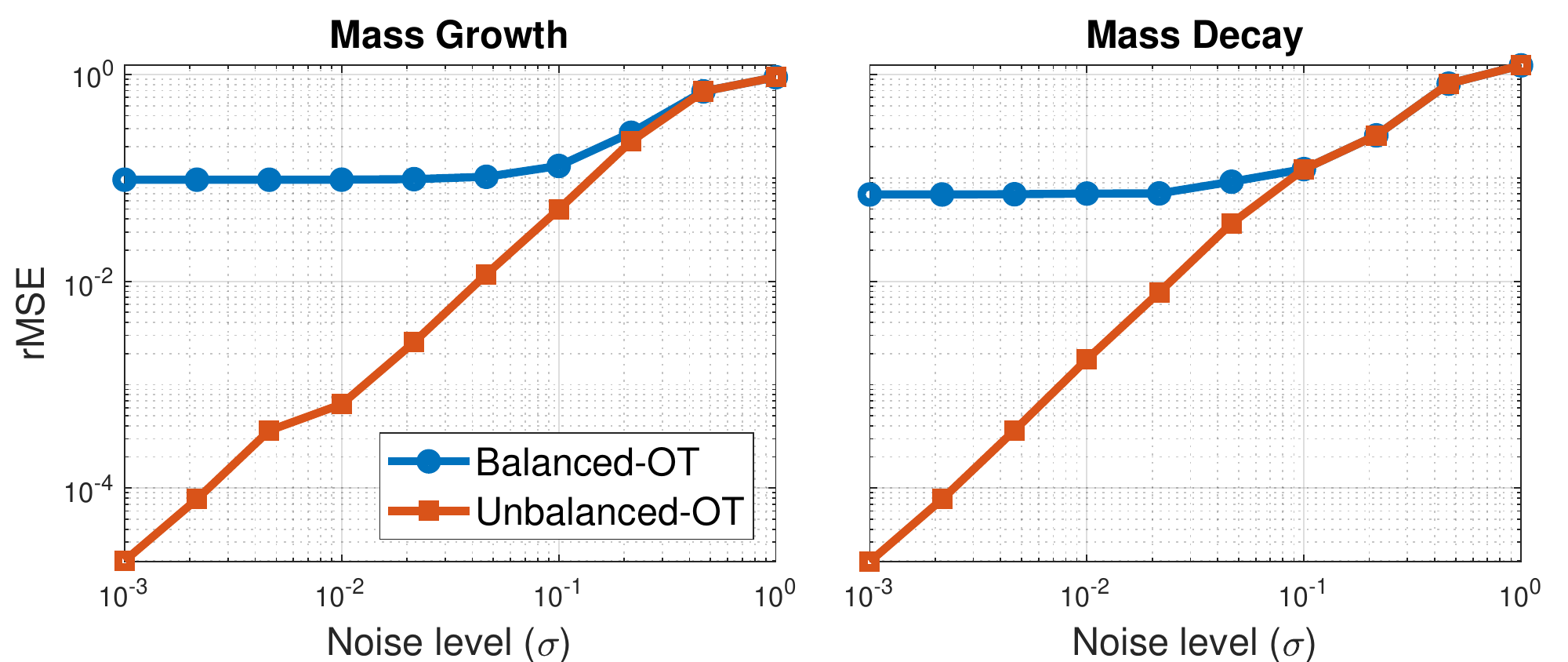}}
  \vspace{0.25cm}
  \centerline{\includegraphics[trim={0cm 0cm 0cm 0.65cm},clip,width=0.49\textwidth]
  {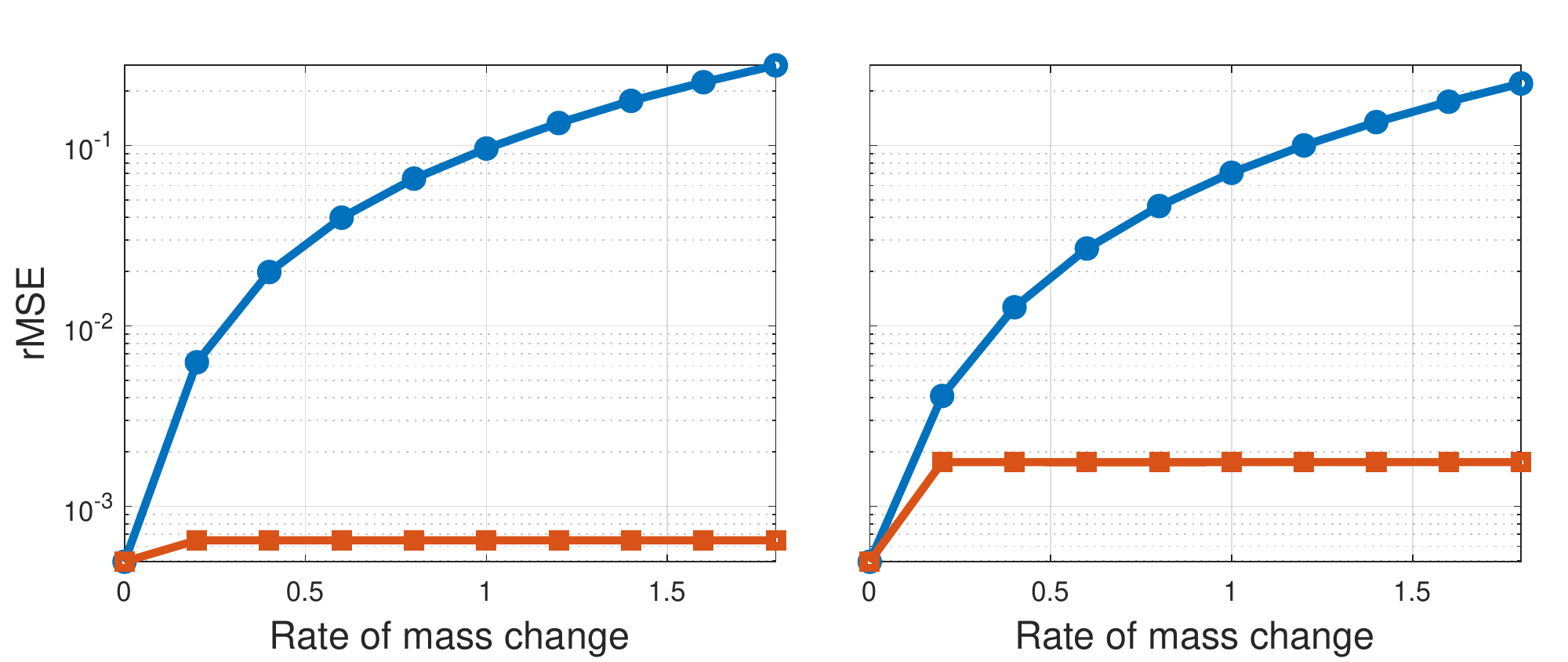}}
\caption{
  Quantitative reconstruction performance of OT-regularization schemes at varying noise levels (first row) and at varying growth/decay rates (second row).
  We subdivide the experiments into categories of mass growth (first column) and mass decay (second column).
  Overall, unbalanced-OT has superior performance across various noise-levels and across various rates of mass change: unbalanced-OT has a reconstruction error that empirically lower bounds that of balanced-OT.
  Especially in mass changing regimes, unbalanced-OT exhibits significant robustness compared to balanced-OT.
}
\label{fig:007_quantitative}
\end{figure}

In Figure \ref{fig:007_qualitative}, we inspect their qualitative behaviors.
Under a growth regime, balanced-OT  transports insufficient mass from $\v{s}_0$, resulting in a gross under-estimation of its amplitude.
Conversely, under the decay regime, too much mass is transported so excess mass has to ``overflow'' into neighboring support.
In the growth regime, support estimation is interestingly perfect due to the limited transport budget, while support estimation is poor in the decay regime due to excess transport budget.
For unbalanced-OT, we observe that it has similar solutions (though not identical) under both regimes, with good support estimation and low reconstruction errors.
Reconstruction error is measured using \emph{relative mean squared error} (rMSE) computed as $\norm{\hat{\v{s}}-\v{s}^\star}_2^2 / \norm{\v{s}^\star}_2^2$ for a given estimation $\hat{\v{s}}$ and ground truth $\v{s}^\star$; rMSE lies in the range $[0,\infty]$ with 0 representing perfect reconstruction.
Unbalanced-OT has a mechanism that adjusts its transported mass (via parameter $\mu$) so it does not suffer from mass overflow or underflow issues of balanced-OT.
Furthermore, unbalanced-OT enjoys favorable support estimation because it allows mass growth/decay at the individual pixel level.
In Figure \ref{fig:007_quantitative}, we quantitatively compare the reconstruction error of the two OT-regularizers.
In the first set of plots (top row) we vary noise levels (under a fixed mass-change rate of $0.5$), and in the second set of plots (bottom row) we vary mass-change rates (under a fixed noise level of $\sigma=0.1$).
Across all noise levels and mass-change rates, unbalanced-OT consistently achieves the best performance.


\subsubsection{Comparisons to sparse tracking algorithms in an occlusive setting}
\label{sssec:results_literature_comparison}

To highlight advantages of our method, we compare BPDN+UOT-DF \eqref{eq:BPDN+UOT-DF} against two methods in the sparse tracking literature.
Specifically, we compare
(i) BPDN with an $\ell_1$-dynamic regularizer (L1) \cite{vaswani2010ls, sejdinovic2010bayesian, farahmand2011doubly, Charles2016DynamicFiltering} that promotes temporal continuity by penalizing the $\ell_1$ difference between the current and previous frames, and
(ii) reweighted $\ell_1$ dynamic-filtering (RWL1) \cite{Charles2016DynamicFiltering} that robustly promotes temporal continuity by propagating higher order statistics in Markovian fashion.

This simulation differs from the previous one of Figure~\ref{fig:007_qualitative} in two ways.
First, we consider the more realistic scenario where a longer video are estimated in an online fashion. Here, the reconstructed estimate of the previous frame serves as the signal prior for the next frame; only the prior to the first frame is seeded with the ground truth.
Second, we introduce an occlusive region in the middle of each frame to study how various algorithms cope with extremely sharp temporal discontinuities when active pixels `appear/disappear' from behind the occlusion. Individual pixel intensities are otherwise kept constant across frames.

In Figure~\ref{fig:010}, we illustrate using a representative simulation that UOT produces the best reconstructions both quantitatively (lowest rMSE) and qualitatively.
Compared to the other methods, UOT's flat and smooth quantitative error profile reflects its ability to effectively smooth temporal information.
This property provides robustness against isolated corrupted frames that are reflected as error spikes in the competing methods.
Although BOT also initially demonstrates a similar smoothing ability, its reconstruction is no longer accurate when the total observed mass contracts (after frame 6).

\begin{figure}[htb]
  \centering
  \centerline{\includegraphics[trim={0cm 0cm 0cm 0cm},clip,width=0.49\textwidth]
  {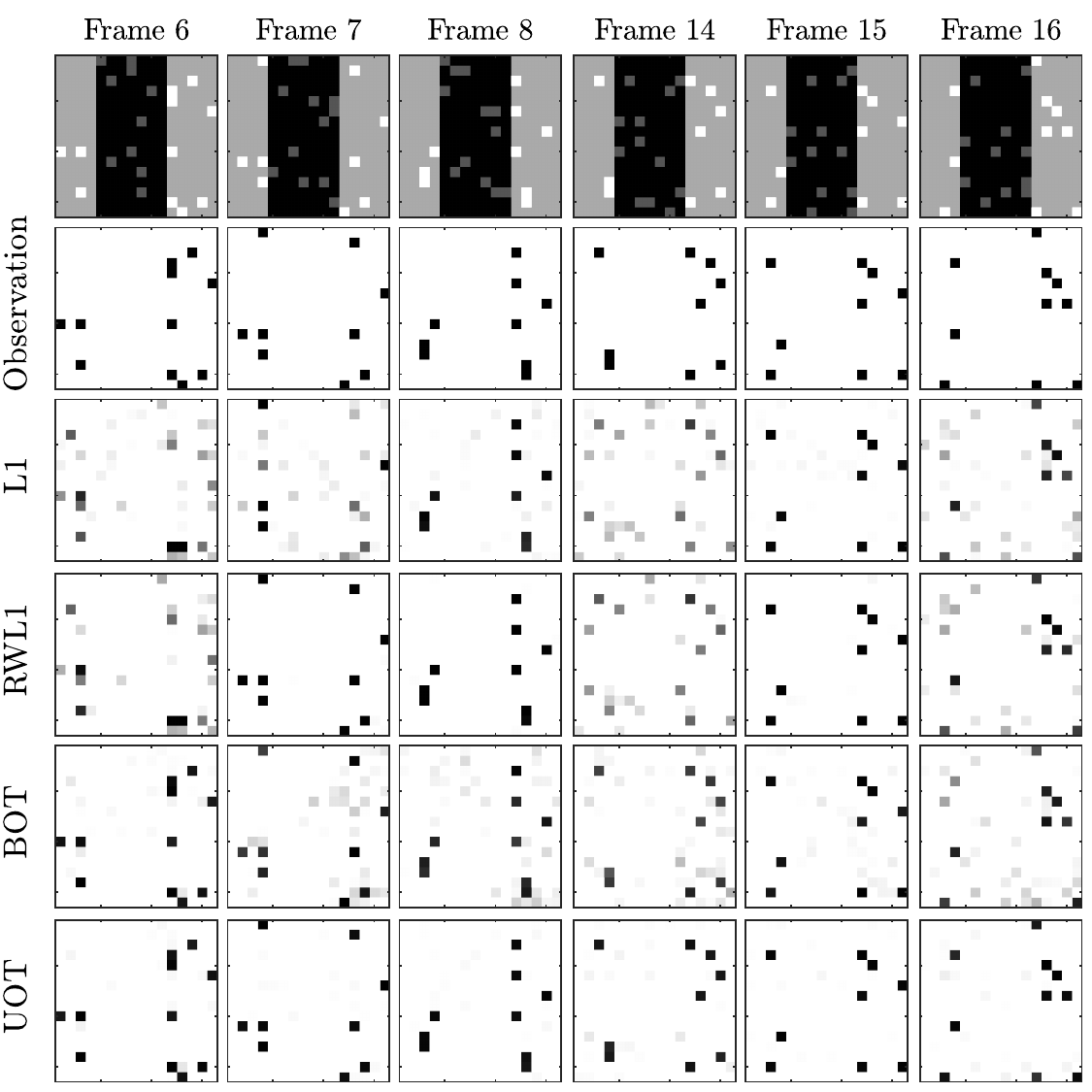}}
  \vspace{0.25cm}
  \centerline{\includegraphics[trim={0.6cm 0.1cm 0.9cm 0.2cm},clip,width=0.49\textwidth]
  {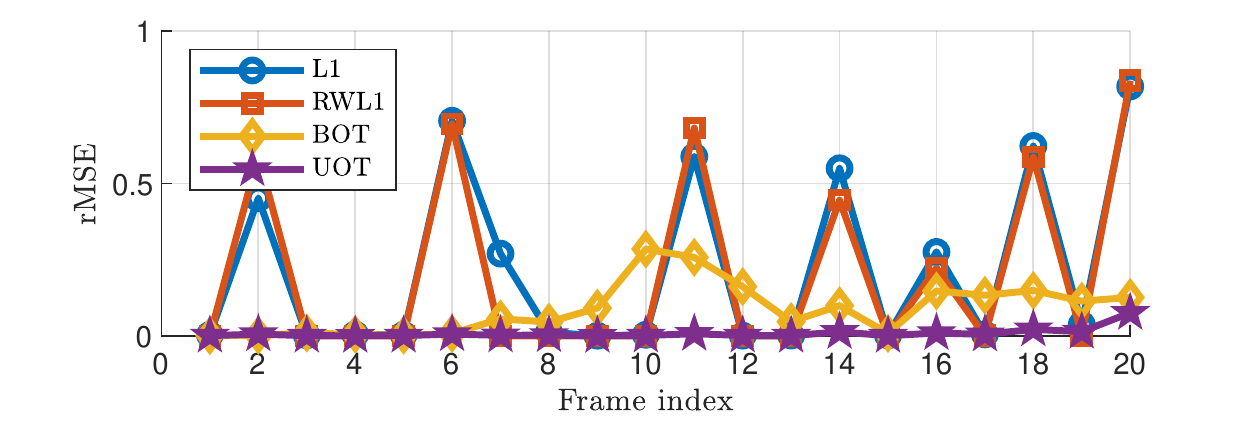}}
\caption{
  Qualitative and quantitative reconstruction performance of BPDN augmented with various dynamical priors in an occlusive setting.
  The top row displays frames of random target motion in a scene that contains a vertical (darkened) occlusive region (only light gray regions are visible to the observer).
  The second row shows the (uncompressed) observations whose occlusive regions have been subtracted (before noise is added).
  The following rows show a sample of reconstructed frames from four sparse tracking methods: BPDN with unbalanced-OT (UOT), BPDN with balanced-OT (BOT), BPDN with $\ell_1$ (L1), and reweighted $\ell_1$ dynamical filtering (RWL1).
  The quantitative plot shows the per-frame reconstruction using the relative mean squared error (rMSE) metric.
  Qualitatively and quantitatively, UOT is the superior method.
}
\label{fig:010}
\end{figure}


\subsection{Robust PCA with UOT-regularization}
\label{ssec:RPCA_batch_setting}

For our second application, we consider an ill-posed joint separation and inverse problem where the goal is to separate a superposition of signal $\v{S}$ and clutter (or background/interference) $\v{L}$ embedded in noisy measurements $\v{Y}$.
A similar dynamical propagation model (as previously described) is assumed between consecutive frames of $\v{S} = [\v{s}_1, \ldots,\v{s}_T] \in \reals^{N \times T}$, while clutter is denoted with $\v{L} = [\v{l}_1,\dots,\v{l}_T] \in \reals^{N \times T}$ where $\rank(\v{L}) = R \leq \min(N,T)$.
Measurements $\v{Y} = [\v{y}_1,\dots,\v{y}_T] \in \reals^{M \times T}$ are captured by the model
\begin{equation*}
  \v{y}_t = \v{\Phi}_t ( \v{s}_t + \v{l}_t ) + \v{\eta}_t ,
  \qquad
  t = 1, \dots, T ,
\end{equation*}
where $\v{\Phi}_t \in \reals^{M \times N}$ is the measurement matrix, and $\v{\eta}_t \in \reals^{M}$ denotes iid Gaussian measurement noise.
Unlike the online nature of \eqref{eq:UOT-DF}, this application  requires a batch of time-windowed data to capture sufficient information so that clutter can be differentiated from the sparse signals.
To solve this problem, we employ a framework called \textit{robust principal components analysis (RPCA)} to separate $\v{S}$ from $\v{L}$, and augment it using our unbalanced-OT model to exploit dynamical continuity between the adjacent frames of $\v{S}$.

RPCA separates data $\v{X}$ comprising of a superposition of sparse outliers $\v{S}$ and a low-rank component $\v{L}$ via
\begin{equation} \label{eq:RPCA_pure}
  \min_{\v{S},\v{L}} \enspace
  \|\v{S}\|_1 + \gamma \|\v{L}\|_\ast
  \enspace \text{s.t.} \enspace
  \v{X} = \v{S} + \v{L} ,
\end{equation}
where $\|\v{S}\|_1 = \sum_{ij} |S_{ij}|$ and $\|\v{L}\|_\ast$ refers to the nuclear norm, (i.e., sum of magnitudes of $\v{L}$'s singular values) with some parameter $\gamma > 0$.
It was shown in \cite{candes2011robust, wright2013compressive} that under incoherence and randomness conditions on $\v{S}$ and $\v{L}$, solving \eqref{eq:RPCA_pure} with parameter $\gamma = \sqrt{\max(N,T)}$ recovers $\v{S}$ and $\v{L}$ with high probability, provided $\v{S}$ is sufficiently sparse and $\v{L}$ is sufficiently low-rank.
A relevant application to this work is video surveillance \cite{candes2011robust}, where video frames are packed into the columns of $\v{X}$.
In this scenario, $\v{L}$ corresponds to the stationary background while $\v{S}$ captures moving objects in the foreground.
A rich literature has developed around the idea of sparse and low-rank decompositions \cite{Bouwmans2016HandbookRobustLowrank}. One branch of work focuses on enforcing additional structure on the sparse component to encourage solutions that vary continuously over time, for example by employing optical flow based methods \cite{Ye2015ForegroundBackgroundSeparation} or Markov Random Fields \cite{Zhou2013MovingObjectDetection}. However, these methods assume full access to the original video frames and would require non-trivial adaptations to allow recovery in the context of an inverse problem.

In this work, we extend the stable formulation of RPCA \cite{Zhou2010StablePrincipalComponent} with linear measurements \cite{Waters2011SpaRCSRecoveringLowrank}
\begin{equation} \label{eqn:stable_RPCA}
\begin{aligned}
  \min_{\v{S} , \v{L} \geq 0} \enspace
  &\frac{1}{2} \sum_{t=1}^T \big( \| \v{y}_t - \v{\Phi}_t (\v{s}_t + \v{l}_t) \|_2^2 \big) \\
  &+ \lambda \| \v{S} \|_1
  + \gamma \| \v{L} \|_\ast,
\end{aligned}
\end{equation}
and show how our unbalanced-OT model can easily be incorporated to use optimal transport as a continuity regularizer on the sparse component.
Specifically, we propose robust PCA with UOT-regularized dynamic-filtering (RPCA+UOT-DF):
\begin{equation} \label{eq:RPCA_UOTDF}
\begin{aligned}
  \min_{\v{S} , \v{L} \geq 0} \enspace
  &\frac{1}{2} \sum_{t=1}^T \big( \| \v{y}_t - \v{\Phi}_t (\v{s}_t + \v{l}_t) \|_2^2 \big) \\
  &+ \lambda \| \v{S} \|_1
  + \gamma \| \v{L} \|_\ast
  + \kappa \sum_{t=1}^{T-1} \Vbeck_\mu ( \v{s}_{t} , \v{s}_{t+1} ) .
\end{aligned}
\end{equation}
The first term is a data fidelity term, the second term a sparsity prior, the third term a low-rank prior, and the last term the unbalanced-OT regularizer that promotes temporal coherence across the sparse frames of the signal.
Applying the theory in \cite{wright2013compressive} allows us to reduce one parameter due to the relation $\gamma/\lambda = \sqrt{\max(N,T)}$.

\subsubsection{ADMM solver for RPCA+UOT-DF}
\label{sssec:RPCA_UOTDF_ADMM}



RPCA+UOT-DF may be interpreted as a massive variational OT problem where $T-1$ OT problems of size $N$ are simultaneously solved: this has a traditional per-iteration computational complexity of $\bigO{N^3 \log N}$.
To make this problem tractable, we propose a proximal first order method based on ADMM \cite{boyd2011distributed} to solve \eqref{eq:RPCA_UOTDF}, and in so doing, highlight the superior $\bigO{N}$ complexity of our unbalanced-OT Beckmann proximal algorithm.
To begin, we perform variable splitting on \eqref{eq:RPCA_UOTDF} to formulate an equivalent problem
\begin{equation*} \label{eq:RPCA_UOTDF_split}
\begin{aligned}
  \min \enspace
  &\frac{1}{2} \sum_{t=1}^T \big( \| \v{y}_t - \v{\Phi}_t \v{x}_t \|_2^2 \big)
  + \lambda \| \v{S} \|_1
  + \gamma \| \v{T} \|_\ast \\
  &+ \ind{+}( \v{S} )
  + \ind{+}( \v{L} )
  + \kappa \sum_{t=1}^{T-1} \Vbeck_\mu ( \v{z}_{t} , \v{w}_{t+1} ) \\
  \text{s.t.} \quad
  &\v{X} = \v{S} + \v{L} , \enspace
  \v{T} = \v{L} , \\
  &\v{z}_t = \v{s}_t , \enspace
  \v{w}_{t+1} = \v{s}_{t+1} , \enspace t = 1,\dots,T-1,
\end{aligned}
\end{equation*}
where $\v{X},\v{T},\v{Z},\v{W}$ are auxiliary variables.
Introducing multiplier variables $\v{A},\v{B},\v{C},\v{D}$, the \emph{augmented} Lagrangian is
\begin{equation*}
\begin{split}
  \mathcal{L}(\v{S},\v{L},\v{X},\v{T},\v{Z},\v{W},\v{A},\v{B},\v{C},\v{D}) =
  \frac{1}{2} \sum_{t=1}^T \big( \| \v{y}_t - \v{\Phi}_t \v{x}_t \|_2^2 \big) \\
  + \lambda \| \v{S} \|_1
  + \gamma \| \v{T} \|_\ast
  + \ind{+}( \v{S} )
  + \ind{+}( \v{L} )
  + \kappa \sum_{t=1}^{T-1} \Vbeck_\mu ( \v{z}_{t} , \v{w}_{t+1} ) \\
  + \frac{\rho}{2} \Big( \norm{\v{X} - \v{S}-\v{L} + \v{A}}_F^2 + \norm{\v{L} - \v{T} + \v{D}}_F^2 \\
  +                      \sum_{t=1}^{T-1} \norm{\v{z}_t - \v{s}_t + \v{b}_t}_2^2 + \norm{\v{w}_{t+1} - \v{s}_{t+1} + \v{c}_{t+1}}_2^2 \Big) ,
\end{split}
\end{equation*}
noting that matrices are expressed with upper-case bold letters while lower-case bold letters refer to their respective columns (with index denoted by subscript $t$).
ADMM generates a convergent sequence of updates by successively solving for each variable per iteration.
The updates are as follows.
\begin{equation*}
\begin{aligned}
  \v{s}_t \iternext
  & \leftarrow \argmin{\v{s} \geq 0}
    \lambda \| \v{s} \|_1
  + \frac{\rho\sigma_t}{2} \norm{\v{s} - \v{k}\iternext_t / \sigma_t }_2^2 \\
  & = \shrink^{\ell_1}_{\lambda/\rho\sigma_t} ( \v{k}\iternext_t / \sigma_t ) ,
\end{aligned}
\end{equation*}
where $\v{k}\iternext_t = (\v{x}\itercur_t - \v{l}\itercur_t + \v{a}\itercur_t) + \omega^z_t (\v{z}\itercur_t + \v{b}\itercur_t) + \omega^w_t (\v{w}\itercur_t + \v{c}\itercur_t)$ and $\sigma_t = 1 + \omega^z_t + \omega^w_t$, with $\omega^z_t = \omega^w_{t+1} = 1$ for $t = 1,\dots,T-1$ and $0$ otherwise.
\begin{equation*}
  \v{L} \iternext
    \leftarrow \argmin{\v{L} \geq 0}
    \norm{\v{L} - \v{K}\iternext }_F^2
  = \Pi_+ ( \v{K}\iternext ) ,
\end{equation*}
where $\v{K}\iternext = \frac{1}{2} (\v{T}\itercur-\v{D}\itercur + \v{X}\itercur-\v{S}\iternext+\v{A}\itercur)$.
\begin{equation*}
\begin{aligned}
  \v{x}\iternext_t
  & \leftarrow \argmin{\v{x}}
    \norm{\v{y}_t - \v{\Phi}_t \v{x}}_2^2
  + \rho \norm{\v{x} - \v{k}\iternext_t }_2^2 \\
  & = (\v{\Phi}^\top \v{\Phi} + \rho\v{I})^{-1} ( \v{\Phi}^\top_t\v{y}_t + \v{k}\iternext_t ) ,
\end{aligned}
\end{equation*}
where $\v{k}\iternext_t = \v{s}\iternext_t + \v{l}\iternext_t - \v{a}\itercur_t$.
\begin{equation*}
\begin{aligned}
  \v{T}\iternext
  & \leftarrow \argmin{\v{T}}
    \gamma \norm{\v{T}}_\ast
  + \frac{\rho}{2} \norm{\v{T} - (\v{L}\iternext+\v{D}\itercur) }_2^2 \\
  & = \shrink^{\ast}_{\gamma/\rho} ( \v{L}\iternext+\v{D}\itercur ) ,
\end{aligned}
\end{equation*}
where the singular value thresholding operator is given by $\shrink^{\ast}_{\sigma} (\v{X}) = \v{U} \shrink^{\ell_1}_{\sigma}(\v{\Sigma}) \v{V}^\top$, where $\v{X} = \v{U}\v{\Sigma}\v{V}^\top$ is any singular-value decomposition.
\begin{equation*}
\begin{aligned}
  \begin{bmatrix} \v{z}\iternext_t \\ \v{w}\iternext_{t+1} \end{bmatrix}
  & \leftarrow \argmin{\v{z},\v{w}}
  \Biggl(
    \kappa \Vbeck_\mu (\v{z},\v{w}) \\
  &\qquad\qquad\qquad + \frac{\rho}{2} \norm{ \begin{bmatrix} \v{z} \\ \v{w} \end{bmatrix} - \begin{bmatrix} \v{s}\iternext_t - \v{b}\itercur_t \\ \v{s}\iternext_{t+1} - \v{c}\itercur_{t+1} \end{bmatrix} }_2^2 \Biggr) \\
  & = \proximal_{\kappa/\rho\Vbeck}^{\mu} ( \v{s}\iternext_t - \v{b}\itercur_t , \v{s}\iternext_{t+1} - \v{c}\itercur_{t+1} ) .
\end{aligned}
\end{equation*}
\begin{equation*}
\begin{aligned}
  \v{A}\iternext &\leftarrow \v{A}\itercur + \v{X}\iternext - \v{S}\iternext + \v{L}\iternext \\
  \v{D}\iternext &\leftarrow \v{D}\itercur + \v{L}\iternext - \v{T}\iternext \\
  \v{b}\iternext_t &\leftarrow \v{b}\itercur_t + \v{z}\iternext_t - \v{s}\iternext_t, \enspace \forall t = 1,\dots,T-1 \\
  \v{c}\iternext_{t+1} &\leftarrow \v{c}\itercur_{t+1} + \v{w}\iternext_{t+1} - \v{s}\iternext_{t+1}, \enspace \forall t = 1,\dots,T-1 .
\end{aligned}
\end{equation*}

The full algorithm is presented in Algorithm \ref{algo:RPCA_UOTDF}.
Each update has linear run-time complexity with the exception of lines 9 and 10.
The dominant cost per iteration comes from evaluating the singular value thresholding operator in line 10, which can be mitigated by applying either a partial or approximate SVD \cite{candes2011robust}.
Line 9 solves $T$ inverse problems; its complexity depends on how $\v{\Phi}_t$'s are structured ($\bigO{TN}$ if diagonal, $\bigO{TN^2}$ if dense with prefactorization).

\begin{algorithm}
\caption{RPCA+UOT-DF ADMM Algorithm.}
\label{algo:RPCA_UOTDF}
\begin{algorithmic}[1]
\Require $\{\v{y}_t, \v{\Phi}_t\}_{t=1}^T$, $\lambda$, $\gamma$, $\kappa$, $\mu$, $\rho$
\Ensure $\v{S}\iternext$, $\v{L}\iternext$
\State $k = 1$
\State $\v{S}\itercur \leftarrow \v{L}\itercur \leftarrow \v{X}\itercur \leftarrow \v{T}\itercur \leftarrow \v{Z}\itercur \leftarrow \v{W}\itercur \leftarrow \v{A}\itercur \leftarrow \v{B}\itercur \leftarrow \v{C}\itercur \leftarrow \v{D}\itercur \leftarrow \v{0}$
\While {\text{not converged}}
    \State $\omega^z_t \leftarrow \omega^w_{t+1} = 1$ for $t = 1,\dots,T-1$ and $0$ otherwise
    \State $\v{k}\iternext_t \leftarrow (\v{x}\itercur_t - \v{l}\itercur_t + \v{a}\itercur_t) + \omega^z_t (\v{z}\itercur_t + \v{b}\itercur_t) + \omega^w_t (\v{w}\itercur_t + \v{c}\itercur_t)$
    \State $\sigma_t \leftarrow 1 + \omega^z_t + \omega^w_t$
    \State $\v{s}\iternext_t \leftarrow \shrink^{\ell_1}_{\lambda/\rho\sigma_t} ( \v{k}\iternext_t / \sigma_t ), \enspace \forall t=1,\dots,T$
    \State $\v{L}\iternext \leftarrow \frac{1}{2} \Pi_+ ( \v{T}\itercur-\v{D}\itercur + \v{X}\itercur-\v{S}\iternext+\v{A}\itercur )$
    \State $\v{x}\iternext_t \leftarrow (\v{\Phi}_t^\top \v{\Phi}_t + \rho\v{I})^{-1} ( \v{\Phi}^\top_t\v{y}_t + \v{s}\iternext_t + \v{l}\iternext_t - \v{a}\itercur_t )$, $\forall t = 1,\dots,T$
    \State $\v{T}\iternext \leftarrow \shrink^{\ast}_{\gamma/\rho} ( \v{L}\iternext+\v{D}\itercur )$
    \State $( \v{z}\iternext_t , \v{w}\iternext_{t+1} ) \leftarrow \proximal_{\kappa/\rho\Vbeck}^{\mu} ( \v{s}\iternext_t - \v{b}\itercur_t , \v{s}\iternext_{t+1} - \v{c}\itercur_{t+1} )$
    \State $\v{A}\iternext \leftarrow \v{A}\itercur + \v{X}\iternext - \v{S}\iternext + \v{L}\iternext$
    \State $\v{D}\iternext \leftarrow \v{D}\itercur + \v{L}\iternext - \v{T}\iternext$
    \State $\v{b}\iternext_t \leftarrow \v{b}\itercur_t + \v{z}\iternext_t - \v{s}\iternext_t, \enspace \forall t = 1,\dots,T-1$
    \State $\v{c}\iternext_{t+1} \leftarrow \v{c}\itercur_{t+1} + \v{w}\iternext_{t+1} - \v{s}\iternext_{t+1}, \enspace \forall t = 1,\dots,T-1$
    \State $k \leftarrow k + 1$
\EndWhile
\end{algorithmic}
\end{algorithm}

\subsubsection{Efficiency of proximal Beckmann in ADMM implementation of RPCA+UOT-DF}
\label{sssec:prox_efficiency}

Proximal algorithms have maximal utility when they can be computed efficiently (i.e., closed form solutions at best, and iterative at worst).
Although the proposed Beckmann proximal algorithm of Section \ref{ssec:UOTprox} is iterative, we submit that it can still be extremely efficient by applying standard strategies such as:
\begin{itemize}
    \item \textit{warm starts} -- instead of restarting the Beckmann proximal algorithm at each ADMM iteration, we warm-start it using its state from the previous ADMM iteration, and
    \item \textit{inexact updates} (early termination) -- rather than solving the proximal algorithm to high precision, we only partially solve it by terminating it early after a fixed (predetermined) number of iterations.
\end{itemize}
To empirically demonstrate the benefits of the above strategies, we apply them to our proposed Beckmann proximal algorithm within the ADMM framework of the RPCA+UOT-DF solver described in Section \ref{sssec:RPCA_UOTDF_ADMM}.
In Figure \ref{fig:008} we inspect the computational complexity associated with ADMM as a function of solution exactness (i.e., accuracy) of the Beckmann proximal algorithm.
Exactness here is implied by the number of fixed iterations we let it run (i.e., more iterations imply a more exact solution).
To measure the computational complexity of ADMM, we count the number of ADMM iterations required for it satisfy the following stopping criterion: the primal and dual residuals (defined \cite[\S3.3]{boyd2011distributed}) must both reach a value less than $10^{-4}$.
\textit{Additional \% of ADMM iterations} (y-axis) is defined as $|b(k) - b(30)|/b(30)$, where $b(k)$ refers to the recorded number of ADMM iterations at convergence, while $k$ refers to the number of fixed Beckmann proximal iterations (x-axis).
We use $b(30)$ as our baseline since convergence of the proximal algorithm is observed at this value of $k$.
Unsurprisingly, our results demonstrate that that ADMM's computational complexity decays as a function of proximal solution exactness.
However, even at the upperbound (a single proximal iteration), the overall ADMM complexity is inflated by merely $1 \%$. This suggests that our proximal algorithm can be very efficient when used with these strategies.

\begin{figure}[htb]
  \centering
  \centerline{\includegraphics[trim={0cm 0cm 0cm 0cm},clip,width=0.49\textwidth]
  {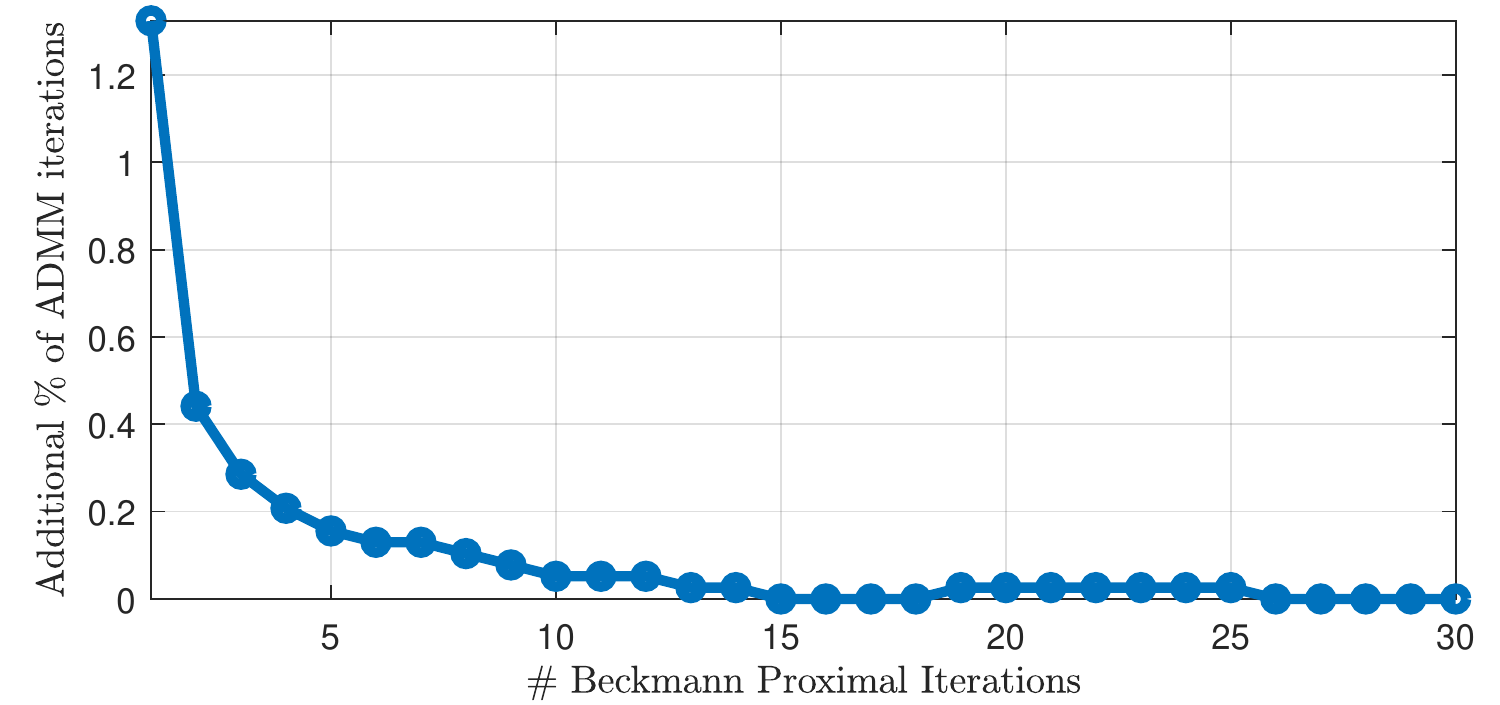}}
\caption{
  Additional ADMM computational cost (expressed as a percentage over the minimum possible number of ADMM iterations) as a function of proximal iterations.
  When warm-starts and early-termination strategies are applied, we see that even in worst case (just 1 proximal iteration), overall costs increase by merely $1 \%$.
}
\label{fig:008}
\end{figure}

\subsubsection{RPCA+UOT-DF performance on synthetic data}
To motivate this section's synthetic data example, consider now a radar scenario where targets are to be  recovered from underdetermined linear measurements and there exists both additive noise as well as interference (e.g., from other transmitters) that must be separated out.
We characterize RPCA+UOT-DF ($p=1$) on synthetic data generated under this scenario using the following procedure.
The matrix $\v{S} = [\v{s}_1, \dots, \v{s}_T]$ is a sparse matrix whose columns $\v{s}_t \in \reals^N_+$ are vectorized $n \times n$ images ($N = n^2$).
We begin by generating $K$ active pixels (targets) in $\v{s}_1$ then, for each consecutive frame, randomly move them in similar fashion as described in Section \ref{sssec:results_comparison_OT_schemes}.
The top row of Figure \ref{fig:000} illustrates how the ground truth is generated on a simple dataset with $10 \times 10 \text{ pixels} \times 6 \text{ frames}$.
$\v{L} = [\v{l}_1, \dots, \v{l}_T]$ is generated by multiplying two low-rank matrices $\v{U}\v{V}^\top/4R$ where $\v{U}\in\reals^{N \times R}_+$ and $\v{V}\in\reals^{T \times R}_+$ are matrices of rank $R \leq \min(N,T)$, and whose entries are distributed as $U_{ij},V_{ij} \sim \mathrm{Uniform}(0,1)$.
$\v{\Phi}_t \in \reals^{M \times T}$ have entries that are randomly generated from an iid normal distribution with a variance of $1/M$.
The observation matrix $\v{Y} = [\v{y}_1,\dots,\v{y}_T]$ consists of columns produced by $\v{y}_t = \v{\Phi}_t ( \v{s}_t + \v{l}_t ) + \v{\eta}_t$, with noise $\v{\eta}_t \sim \mathcal{N}(0,\sigma^2 \v{I})$.

In our experiments, we compared against three other algorithms: (i) robust PCA (RPCA) \cite{candes2011robust} to serve as a benchmark, (ii) RPCA with a \emph{balanced-OT} regularizer (RPCA+BOT-DF), and (iii) RPCA with an $\ell_1$ dynamical filter (RPCA+L1-DF) as a cheap dynamical filtering alternative to OT methods.
Since there currently exists no sparsity-based dynamically-aided RPCA competitors in the literature, these algorithms act as hypothetical contenders.
For this solver, we replace the unbalanced-OT prior in \eqref{eq:RPCA_UOTDF} with $\sum_{t=1}^{T-1} \| b(\v{s}_t) - b(\v{s}_{t+1}) \|_1$ where $b(\cdot)$ refers to a Gaussian blurring convolution operator with a $3 \times 3$ kernel since we do not expect targets to travel too far in this simulation.
For RPCA+BOT-DF, we replace the unbalanced-OT prior in \eqref{eq:RPCA_UOTDF} with a \emph{balanced-OT} prior \eqref{eq:Unbal_OT_beckman}.
All methods are implemented in ADMM \cite{boyd2011distributed} with the following stopping criteria: primal and dual residuals must both be $\leq 10^{-4}$ or a maximum of $5000$ iterations are reached.

In terms of performance metrics we used: (i) \emph{relative MSE} (rMSE) to measure the normalized $\ell_2$ reconstruction error, and (ii) \emph{F1 score} to measure accuracy of estimated support.
The F1 score (also called the \emph{S{\o}rensen-Dice coefficient}) is computed as the harmonic mean between precision and recall of reconstruction and ground truth masks (using a threshold set at an intensity of 0.05) and computed as $2(\norm{\v{m}^\star \cap \hat{\v{m}}}_0)/(\norm{\v{m}^\star}_0+\norm{\hat{\v{m}}}_0)$, where $m_i = 1$ if $x_i \geq \text{threshold}$ and $0$ otherwise, and $\norm{\cdot}_0$ refers to the cardinality of the argument.
F1 score lies between 0 and 1, with 1 representing perfect support estimation.
Experimental plots show aggregated results from 20 randomly generated trials, with the markers displaying the median and error bars denoting the $25^{\text{th}}$ to $75^{\text{th}}$ quantiles.

\begin{figure}[htb]
  \centering
  \centerline{\includegraphics[trim={0cm 0cm 0.75cm 0cm},clip,width=0.475\textwidth]
  {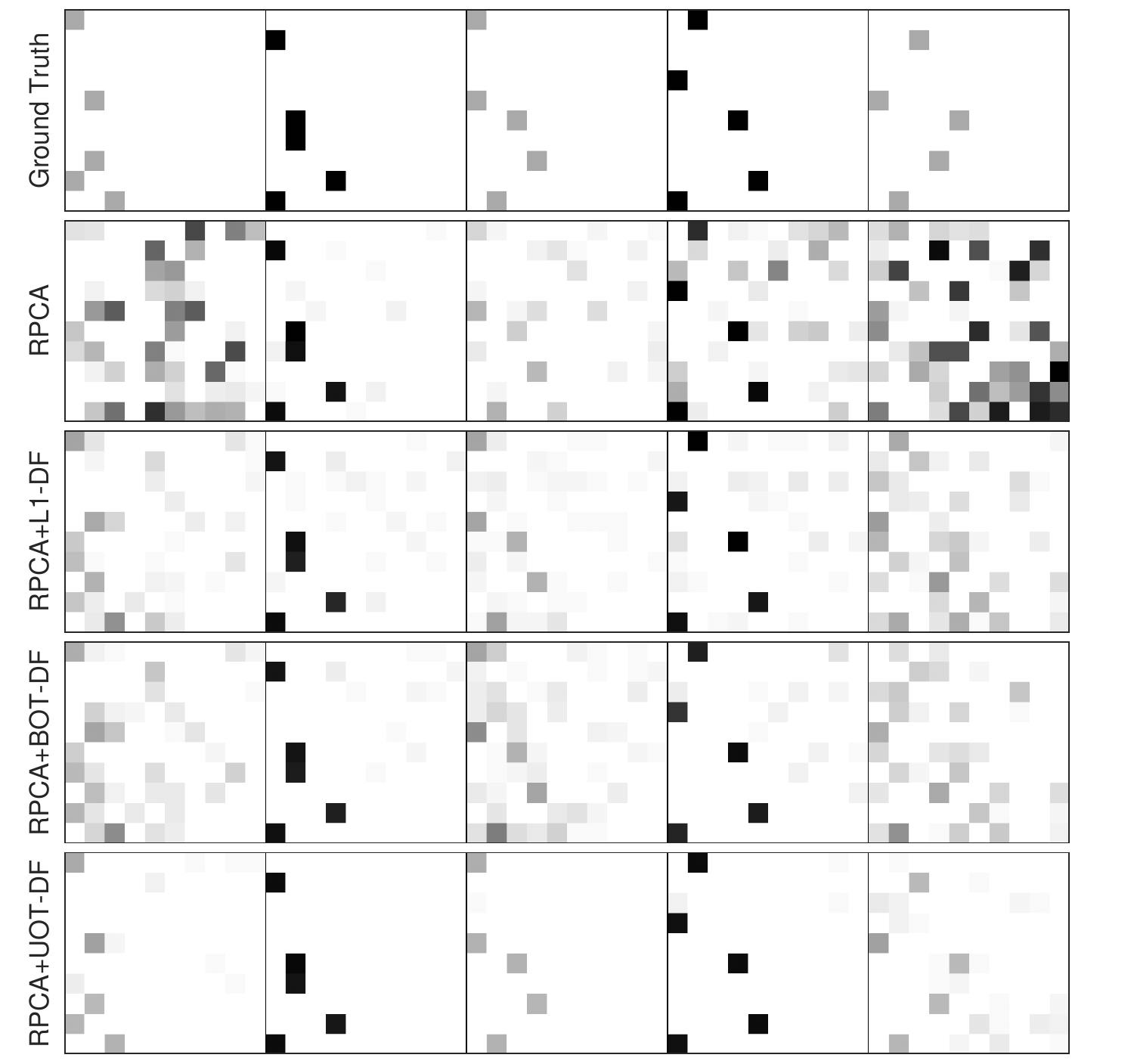}}
\caption{
    Qualitative example of reconstruction of sparse targets from compressive measurements.
    Row labels starting from top: ground truth, robust PCA's reconstruction, RPCA+L1-DF's reconstruction, RPCA+BOT-DF's reconstruction, (proposed) RPCA+UOT-DF's reconstruction.
    While robust PCA is quite successful at reconstructing the sparse targets, it is prone to severe error (frames 1, 3, 4, 5).
    Incorporating an optimal transport regularizer enforces temporal consistency across frames, significantly improving reconstruction performance, especially in accuracy of support.
}
\label{fig:000}
\end{figure}

\begin{figure*}[htb]
  \centering
  \centerline{\includegraphics[trim={0.25cm 2.0cm 0.25cm 0cm},clip,width=1.0\textwidth]
  {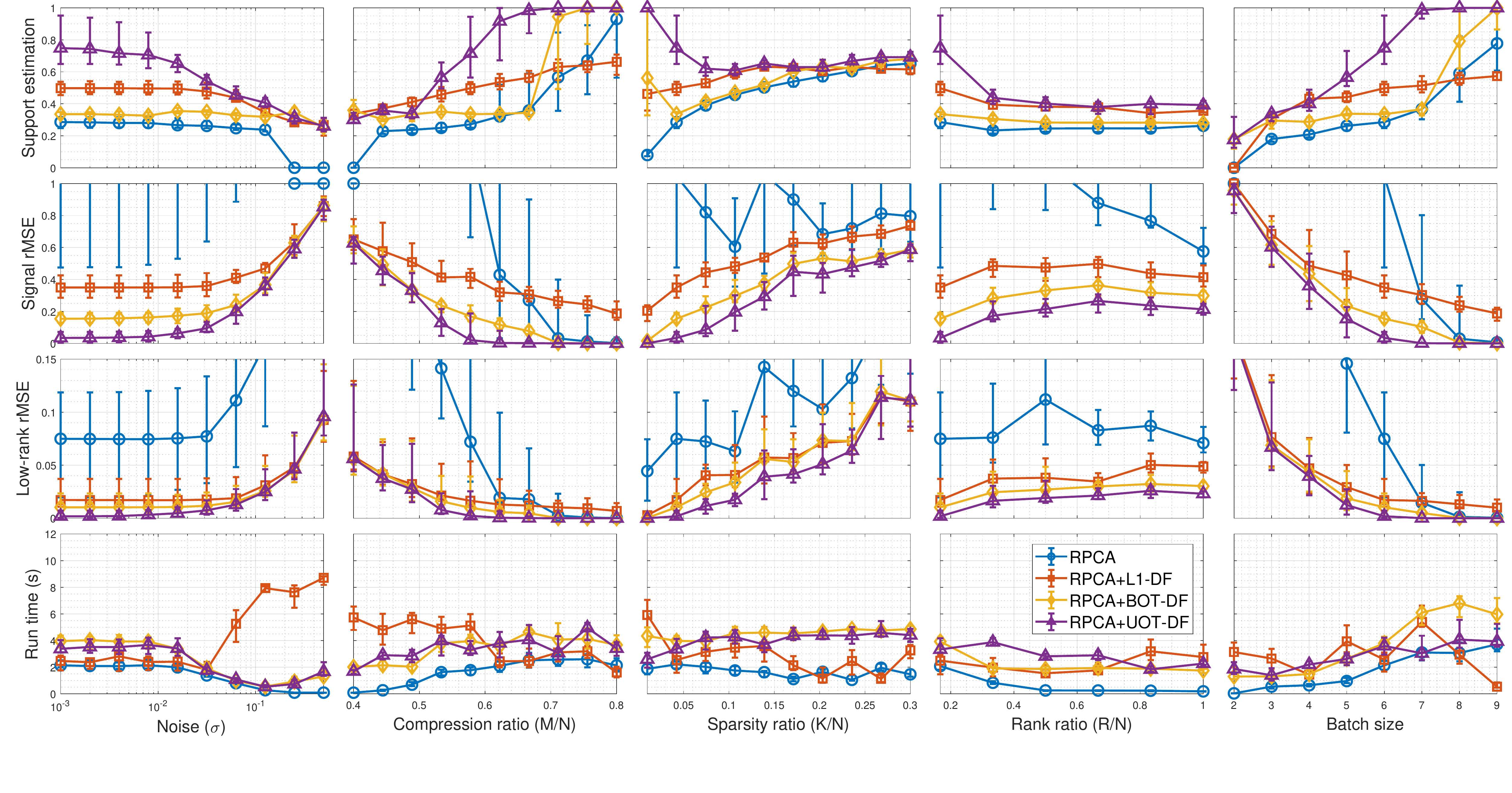}}
\caption{
    Performance of proposed algorithm (RPCA+UOT-DF) compared against other RPCA algorithms on synthetically generated data.
    To understand the algorithm's performance, we vary experimental simulation parameters such as: observation noise (first column), compression rate (second column), number of sparse targets (third column), rank of interference component (fourth column), number of frames in a batch (fifth column).
    The empirical performance limit of all algorithms is discovered by optimizing each algorithm to minimize reconstruction error from 20 independent random trials, as measured by its relative mean square error (rMSE).
    In all plots, we report the metric's trial median using a marker, while error bars denote the $25^{\text{th}}$ to $75^{\text{th}}$ quantiles.
    Our algorithm's main advantage is its significantly superior target support estimation (first row of plots), as measured using F1 score of target masks.
    The proposed algorithm also demonstrates superior performance under the metrics such as rMSE of the sparse component (second row), and rMSE of the low-rank component (third row).
    Finally, we note that algorithm's superior performance is obtained at comparable computational cost to the balanced-OT regularizer (fourth row).
}
\label{fig:001}
\end{figure*}

Figure \ref{fig:000} qualitatively compares the application of our algorithm against the benchmark RPCA.
RPCA generally experiences more noise throughout the reconstruction, and occasionally gross corruption (e.g., frames 1 and 5).
RPCA+L1-DF's sparsity prior on frame differences provides some advantage in support estimation (frame 3 is improved compared to RPCA), yet it is still prone to gross errors (frame 5).
RPCA+BOT-DF similarly shows improvement over RPCA, however we observe mass-overflow effects that were previously also observed in Section \ref{sssec:results_comparison_OT_schemes}.
In contrast, RPCA+UOT-DF demonstrates remarkable robustness and demonstrates significantly improved support estimation over the other algorithms.

Next, we run a series of experiments that characterize the performance of the proposed algorithm in comparison to other algorithms.
To this end, the following simulation parameters are varied:
(i)   the Gaussian noise in the measurements,
(ii)  the compression ratio $M/N$ of the measurement matrix $\v{\Phi}\in\reals^{M \times N}$,
(iii) the number of sparse targets $K$
(iv)  the rank of the low-rank interference component $\v{L}$, and
(v)   the number of frames in a batch $T$.
For each each simulation setting, all solvers are tuned by searching for the optimal parameters (e.g., $\lambda,\kappa,\mu$) via a (logarithmic space) pattern search algorithm \cite{davidon1991variable}, with the low-rank weight assigned as $\gamma = \lambda\sqrt{\max(N,T)}$.
Unless explicitly varied in each experiment, our baseline parameters are: noise-level $\sigma = 0.001$, compression ratio $M/N=0.6$, sparsity of $K=5$ pixels, clutter rank of $R=1$, and rate of mass change of $1$ over the intensity range of $0.5$ to $1.5$.
In figure \ref{fig:001}, we see that the proposed algorithm has superior performance over the other algorithms over all three metrics, especially in its ability to accurately recover the support.
In the first column of plots, we observe that unbalanced-OT is most effective when noise is low ($\sigma < 0.1$), since reconstruction suffers when noise rather than signal is the subject of transportation.
In the second column of plots, we observe a region of significant support accuracy advantage ($M/N < 0.5$) despite marginal rMSE advantage over other algorithms.
The third column of plots showcases unbalanced-OT's advantage in sparse regimes: it exploits sparsity to register targets and provide regularization structure.
In the fifth column of plots, we observe that unbalanced-OT has superior sample complexity, since dynamical correlations are more accurately described by our model.
Finally, we see that although RPCA+BOT-DF and RPCA+UOT-DF shares similar formulations, and similar solve times (fifth row), the simple unbalanced modification obtains a significant performance gain.

\subsubsection{Qualitative performance on real video data}
\label{ssec:videodata}

Finally, we evaluate performance on a video sequence of a person walking through an indoor scene.
The snippet consists of 2 seconds of footage recorded at 30 frames per second and is downsampled to a resolution of $95 \times 160$ pixels ($T=60$ and $N=15,200$).
We highlight the fact that OT-regularized problems of this size were simply intractable before our proposed proximal method.
In such practical applications, the foreground component may be darker or lighter than the background, so we must modify the RPCA+UOT-DF formulation to remove the $\v{S} \geq 0$ constraint.
Since the OT formulation takes nonnegative signals as its inputs, we decompose the sparse component into positive and negative components (in similar fashion as \cite{bertrand2018earth, janati2018wasserstein, profeta2018heat, mainini2012description}) $\v{S} = \v{S}^+ - \v{S}^-$
with $\v{S}^+ , \v{S}^- \geq 0$ and add an OT regularization term for each component.
The RPCA+UOT-DF objective function then becomes
\begin{equation} \label{eq:RPCA_UOTDF_PN}
\begin{aligned}
  \min_{\v{S}^+, \v{S}^- , \v{L} \geq 0} \enspace
  &\frac{1}{2} \sum_{t=1}^T \big( \| \v{y}_t - \v{\Phi}_t (\v{s}^+_t -\v{s}^-_t + \v{l}_t) \|_2^2 \big) \\
  &+ \lambda ( \| \v{S}^+ \|_1 + \| \v{S}^- \|_1 )
  + \gamma \| \v{L} \|_\ast \\
  &+ \kappa \sum_{t=1}^{T-1} \Big( \Vbeck_\mu ( \v{s}^+_{t} , \v{s}^+_{t+1} )
      + \Vbeck_\mu ( \v{s}^-_{t} , \v{s}^-_{t+1} ) \Big) .
\end{aligned}
\end{equation}

In these simulations, we observe linear random projections measurements (in this case, severely compressed with $M/N=0.15$) and use RPCA \eqref{eqn:stable_RPCA}, RPCA+L1-DF and RPCA+UOT-DF \eqref{eq:RPCA_UOTDF_PN} to extract the moving person from the background scene.
As before, pattern search was employed in the selection of algorithm parameters.
However, to avoid the prohibitive computation time required to optimize directly with the full resolution data, parameters were chosen by first using pattern search on heavily downsampled data to obtain an approximation to the optimal parameter set, and then fine-tuned manually using the original data.
Unlike previous simulations, we found that the relationship $\gamma = \lambda \sqrt N$ did not yield optimal results, so $\lambda$ and $\gamma$ were selected independently.
Figure~\ref{fig:lab} shows several example frames which demonstrate how the UOT regularizer enables successful recovery even after compression.
RPCA misses the foreground almost entirely, while RPCA+L1-DF yields only a crude estimate due to the inability of the $\ell_1$ dynamics regularizer to effectively capture continuity in the sparse component.
Similarly, RPCA+BOT-DF fails to cope with the changing mass as the subject enters the frame and walks through the scene.

\begin{figure*}[htb]
  \centering
  \centerline{\includegraphics[trim={0cm 0.2cm 0.1cm 0.1cm},clip,width=1.0\textwidth]
  {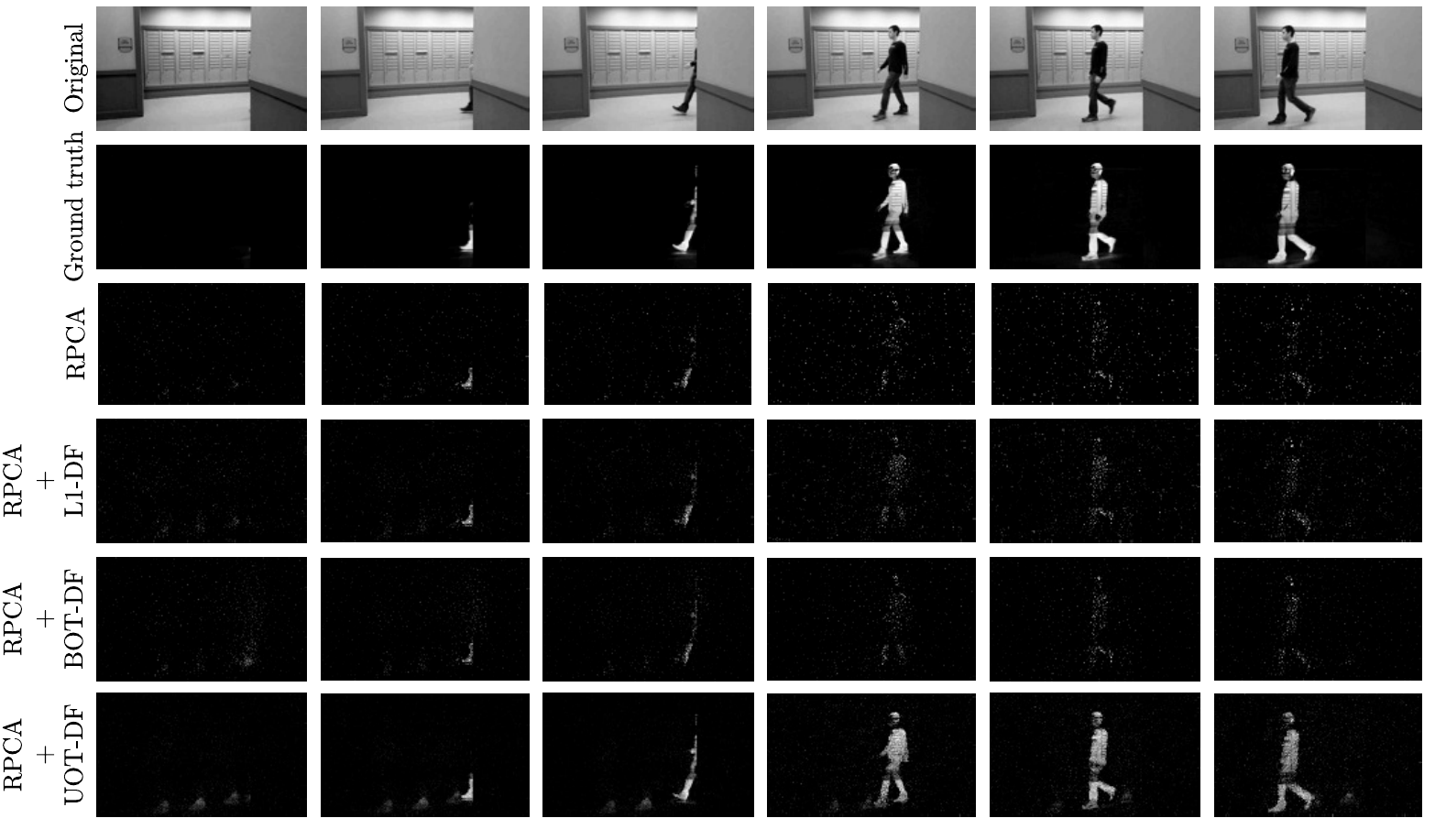}}
\caption{
  Separation of a moving subject from its environment using compressive measurements of a real video clip.
  Even under heavy compression ($M/N=0.15$), the UOT regularizer enables nearly perfect recovery and separation of the person walking through the scene from the background.
  The foreground recovered by RPCA is noisy and the subject is all but lost.
  Although the dynamics regularizer in RPCA+L1-DF reduces noise somewhat, the $\ell_1$ regularizer is not able to effectively leverage the continuity between frames and the subject remains barely distinguishable.
  RPCA+BOT-DF also yields inferior performance due to its constant mass modeling assumption that fails to hold as the subject enters the scene.
}
\label{fig:lab}
\end{figure*}


\section{Summary}

In this paper, we propose a novel regularizer for inverse imaging problems based on recent advances in optimal transport.
We empirically demonstrate that our unbalanced optimal transport regularizer, when cast in Beckmann's formulation, is not only tractable for large scale imaging but also has crucial unbalanced modeling features that overcome limitations of traditional optimal transport constraints.
In addition, we propose a parallel proximal algorithm for this regularizer to allow it to be used with first order solvers for large-scale imaging applications.
To demonstrate the efficacy and efficiency of our method and solver, we focused on target tracking applications in online and batched settings because temporal and spatial continuities are well modelled using our proposed regularizer.
We characterize our method against other benchmarks on a synthetic dataset, and demonstrate superiority at reconstructing dynamical signals.
These promising results indicate that real-time extensions of this work using parallel hardware implementations (e.g., ASIC, FPGA, GPU) will be a fruitful path for future work.
Lastly, we qualitatively demonstrate the efficacy of our method at recovering a moving subject from compressive measurements of real video data.
While our RPCA+UOT method in Section \ref{ssec:videodata} is a successful proof of concept, it may be unable to handle rapid changes in relative darkness/lightness between the foreground and background; for future work, richer OT models should therefore be developed to handle mass transfer between the positive and negative components.



\bibliographystyle{IEEEtran}
\bibliography{refs}

\end{document}